\documentclass[10pt,journal,cspaper,compsoc,twoside]{IEEEtran}

\pdfoutput=1

\usepackage{cite}

\usepackage[cmex10]{amsmath}
\usepackage{array}

\usepackage{fixltx2e}
\usepackage{url}

\usepackage{amssymb}
\usepackage{graphicx}

\newtheorem{theorem}{Theorem}
\newtheorem{definition}{Definition}

\newtheorem{lemma}{Lemma}

\usepackage{comment}
\usepackage{booktabs}
\usepackage{color}

\newcommand{\TODO}[1]{\textcolor{red}{\textbf{****** #1 ******}}}
\renewcommand{\paragraph}[1]{\noindent \textbf{#1}\hspace{1em}}

\usepackage{algpseudocode}

\begin{document}

\title{Bijective Deformations in $\mathbb{R}^n$ via Integral Curve Coordinates}

\author{Lisa~Huynh
        and~Yotam~Gingold
\IEEEcompsocitemizethanks{\IEEEcompsocthanksitem L. Huynh and Y. Gingold are with George Mason University.}
\thanks{}}

\markboth{Bijective Deformations in $\mathbb{R}^n$ via Integral Curve Coordinates}{Bijective Deformations in $\mathbb{R}^n$ via Integral Curve Coordinates}

\IEEEdisplaynotcompsoctitleabstractindextext
\IEEEpeerreviewmaketitle

\IEEEcompsoctitleabstractindextext{%
\begin{abstract}

We introduce Integral Curve Coordinates,
which identify each point in a bounded domain
with a parameter along an integral curve
of the gradient of a function $f$ on that domain;
suitable functions have exactly one
critical point, a maximum,
in the domain,
and the gradient of the function on the boundary points inward.
Because every integral curve intersects the boundary exactly once,
Integral Curve Coordinates provide a natural bijective mapping
from one domain to another given a bijection of the boundary.
Our approach can be applied to shapes in any dimension, provided
that the boundary of the shape (or cage)
is topologically equivalent to an $n$-sphere.
We present a simple algorithm for generating a suitable
function space for $f$ in any dimension.
We demonstrate our approach in 2D
and describe a practical (simple and robust) algorithm for tracing
integral curves on a (piecewise-linear) triangulated regular grid.

\end{abstract}

}

\maketitle

\section{Introduction}

Shape deformation is a widely studied problem in computer graphics,
with applications in
animation, finite element simulation, parameterization,
interactive modeling, and image editing.
Deformations are useful for 2D and 3D shapes, as well as higher-dimensional shapes.
Deforming the animation of a solid model is in fact a 4D problem, as the animating 3D solid is itself a 4D shape.
Interpolating parameterized modeling spaces, such as the space of human body shapes \cite{allen2003space}, can be an arbitrary dimensional problem.

In the version of the problem that we study, a boundary or ``cage'' is created around a shape. As this cage is manipulated, the interior is also deformed.
The cage may be identical to the shape's boundary, which has one fewer dimension
than the shape itself, and is typically more convenient,
as the cage may be simpler (fewer vertices)
or be free of undesirable properties
(such as a non-manifold mesh or high topological genus).

One important yet elusive property for deformations is bijectivity.
A bijective function is a one-to-one mapping such that every point in the undeformed figure is mapped to a point in the deformed figure and vice-versa.
Bijectivity ensures that the shape never flips ``inside-out'' as a result of deformation.


We propose a technique that creates a guaranteed bijective deformation of the shape given a bijective deformation of its boundary or cage, in any dimension. To do so, we introduce Integral Curve Coordinates, which identify each point in the shape by its position along an integral curve of the gradient of a suitable function $f$.
Suitable functions are Lipschitz continuous,
have only a single critical point, a maximum, in the interior,
and have gradients that point inward along the boundary of the domain (or cage).
For such functions, integral curves never meet unless they are identical,
and every integral curve can be uniquely identified with a point on the boundary.
Given a bijective deformation of the boundary points,
we construct a suitable function $f$ and trace the integral curves through the interior
points to extend the bijective deformation into the interior of the shape.
Our approach can be viewed as a generalization of Xia et al.\ \cite{xia2010parameterization} from star-shaped domains to arbitrary contractible domains.

Our contributions are:
\begin{itemize}
\item Integral Curve Coordinates which naturally produce a guaranteed bijective deformation between two contractible domains in $\mathbb{R}^n$, given a bijective deformation of their boundaries.
\item A simple algorithm to create a scalar function space on a
regularly discretized contractible domain in $\mathbb{R}^n$;
functions in this space have exactly one critical point, a maximum.
\item A practical (simple and robust) algorithm for tracing 2D integral curves on a (piecewise linear) triangulated regular grid.
\end{itemize}

We note that the bijections produced by Integral Curve Coordinates, while correct,
are not fair.
The recent work of Sch{\"u}ller et al.~\cite{schuller2013locally}
and Aigerman and Lipman~\cite{aigerman2013injective}
complements our own by improving the fairness of a deformation in a manner that preserves bijectivity.

\section{Related Work}

Deformation is a well-studied problem in computer graphics.
One approach to shape deformation warps the entire ambient space,
which induces a deformation for the coordinates of the shape \cite{Sederberg:1986:FDS,MacCracken:1996:FDL,lipman2012simple}.
In contrast, intrinsic shape deformation approaches operate
in terms of relative coordinates, without regard for the ambient space \cite{Magnenat-Thalmann:1988:JLD,Alexa:2000:ASI,igarashi2005rigid}.
Cage-based deformations enclose the shape
in a sort of structural boundary (which may simply be the shape's actual boundary).
This boundary, when deformed, induces a deformation of the interior points \cite{joshi2007harmonic,lipman2008green,BenChen:2009:VHM,weber2010controllable,nieto2013cage}.
Several approaches allow a combination of cages, line segments, and sparse points while still computing an intrinsic shape deformation~\cite{Sumner:2007:EDS,Jacobson:2011:BBW}.

One approach to cage or boundary deformation involves
the generalization of barycentric coordinates \cite{wachspress1975rational,floater2003mean,ju2005mean,belyaev2006transfinite,joshi2007harmonic,warren2007barycentric,hormann2008maximum,weber2009complex,manson2010moving,weber2011complex}. Traditionally defined, barycentric coordinates express points in the interior of a simplex (triangle in 2D, tetrahedron in 3D) as the weighted average of the vertices of the simplex. The simplex can be thought of as a simple cage around its interior; the deformations induced by modifying vertices of the simplex are bijective (so long as the simplex remains non-degenerate).
Generalized barycentric coordinates extend this idea to more complex polygons or polyhedra than triangles and tetrahedra, but cannot guarantee bijectivity \cite{jacobson2013algorithms}.

\subsection{Bijectivity}

Few deformation approaches guarantee bijectivity. Bijectivity guarantees that the deformation does not invert (flip inside-out) or collapse any part of a shape.
In two-dimensions, Weber et al.~\cite{weber2009complex} introduced conformal mappings based on complex coordinates.
Xu et al.~\cite{Xu:2011:ETG} introduced a thorough solution to the challenging
problem of determining vertex locations for the interior of a triangular graph,
so that no triangles flip inside-out.
This is akin to solving for a bijective deformation
with the additional constraint that mesh connectivity remains unchanged.

While generalized barycentric coordinates cannot
guarantee bijectivity \cite{jacobson2013algorithms},
Schneider et al.~\cite{schneider2013bijective} recently introduced a technique that splits a deformation into a finite number of steps, each of which is implemented via generalized barycentric coordinates. In practice, the technique creates bijective mappings at pixel accuracy. One limitation, however, is that a continuous non-self-intersecting interpolation between the source and target cages is required, which is an unsolved problem for dimensions greater than two.
A theoretical analysis of this technique has not yet been performed in 3D (or higher).

Lipman's restricted functional space, introduced initially in 2D~\cite{lipman2012bounded} and extended to 3D and higher by Aigerman and Lipman~\cite{aigerman2013injective}
can be used to find bijective deformations of piecewise linear meshes similar to a desired one.
The technique projects an input simplicial map onto the space of bounded-distortion simplicial maps.
While they do not prove that their iterative algorithm convergences in general,
if it does converge, it guarantees a bijective deformation.

The approach of Sch{\"u}ller et al.~\cite{schuller2013locally} prevents inverted elements (i.e. guarantees local injectivity) by augmenting any variational deformation approach with a non-linear penalty term that goes to infinity as elements collapse. The approach therefore requires a continuous deformation from an initial shape, and is incompatible with hard constraints (such as a required target pose).

Several approaches have been presented for creating bijective texture maps with positional constraints \cite{kraevoy2003matchmaker,seo2010constrained,lee2008texture}.

The approaches of Weinkauf et al.\ \cite{Weinkauf:2010:TBS}
and Jacobson et al.\ \cite{Jacobson:2012:SSA} both create restricted function spaces
that prevent undesirable extrema (maxima and minima), but cannot control the placement of saddles. Jacobson et al.\ applied their technique to shape deformation.
Our approach creates a restricted function space with exactly one extrema and no saddles, in any dimension.

The goal of our work is to introduce a cage-based deformation technique that guarantees  bijectivity in any dimension, and provide a practical piecewise linear implementation in 2D.

\section{Background}

\begin{definition}
A function $h\colon X \rightarrow Y$ is called bijective if and only if
it is one-to-one:
\[ h(x_1) = h(x_2) \iff x_1 = x_2 \]
and onto:
\[ \forall y \in Y, \; \exists x \in X \text{ such that } h(x) = y. \]
\end{definition}

Bijective functions always have an inverse.
In the literature, the desired bijective deformations are actually homeomorphisms,
which are bijective functions $h$ such that both $h$ and $h^{-1}$ are continuous.

For deformations, the one-to-one property (injectivity) is the most difficult to achieve.
A lack of injectivity manifests as regions of $X$ that collapse or invert (flip ``inside-out'') when mapped via $h$. For piecewise linear shapes, this corresponds
to collapsed or inverted elements (triangles in 2D, tetrahedra in 3D, etc).
Expressed symbolically, a function $h$ must have a Jacobian whose determinant is everywhere positive to be injective:
\begin{equation}
\label{eq:jacobian_positive}
det( J_h( x ) ) > 0 \quad \forall x \in X
\end{equation}
Where the determinant is negative, $h$ has locally flipped $X$ inside out (inverted).
Where the determinant is zero, $h$ has locally collapsed $X$, though not inverted it.
(In 2D, a triangle collapses to a line or point.)

For piecewise linear shapes, the Jacobian is piecewise constant
within each element and can be constructed as follows.
Consider a $k$-dimensional simplex $P$ (triangle in 2D or tetrahedron in 3D).
$P$ has $k+1$ vertices $v_0,v_1,\hdots,v_k$.
Ignoring the overall translation of $h$, the mapping $h(P)$ can be expressed via matrix multiplication.
Choose a vertex of $P$ arbitrarily ($v_0$, without loss of generality).
The matrix we seek maps $v_i - v_0$ to $h(v_i) - h(v_0)$.
Assuming that $P$ is non-degenerate,
$v_1-v_0,\hdots,v_k-v_0$ are linearly independent.
The matrix $M$ such that $M (v_i-v_0) = h(v_i) - h(v_0)$ is, therefore, the product of
two matrices whose columns are vertices:

\begin{equation}
\label{eq:jacobian_matrix_discrete}
\begin{split}
        M&=[ h(v_1) - h(v_0) | h(v_2) - h(v_0) | \cdots | h(v_k) - h(v_0)]\\
        &\times[ v_1 - v_0 | v_2 - v_0 | \cdots | v_k - v_0 ]^{-1}
\end{split}
\end{equation}

The determinant of $M$ determines whether $h$ is locally injective within $P$.
If the determinant is zero, $h$ has locally collapsed $P$ to, in 2D, a line segment or point.
If the determinant is negative, $h$ has inverted $P$.

\section{Overview}
\label{chapter:overview}

Our goal is to deform the interior of a shape given a deformation of its boundary
or enclosing cage.
Formally, our approach takes as input
\begin{enumerate}
    \item The boundary $\partial D$ of a contractible domain $D$ in $\mathbb{R}^n$.
    \item A homeomorphism $h$ deforming $\partial D$.
\end{enumerate}
and extends $h$ to the interior of $D$.

To accomplish this, we introduce Integral Curve Coordinates (formally defined below) that identify a point $\mathbf{p}$ in the domain $D$
with an integral curve $\mathbf{x}(t)$ and a parameter $t_p$ such that $\mathbf{x}(t_p) = \mathbf{p}$.
The integral curves follow the gradient of a special function $f_D$:
$\frac{d}{dt}\mathbf{x}(t) = \nabla f_D( \mathbf{x}(t) )$.
The function $f_D$ is constructed
such that $f_D$ has exactly one critical point in $D$, a maximum,
and its gradient on the boundary
points inward.
Integral curves of the gradient of a function are sometimes called integral lines \cite{Zom01}.
The integral curves of a Lipschitz continuous function are well-defined
everywhere except critical points, and two integral curves
never meet unless they are identical.
Because every integral curve of $\nabla f_D$ traces a path from a point on the boundary
of the domain to the maximum, we can uniquely identify
every integral curve with the boundary point it passes through.
We denote the integral curve passing through boundary point $\mathbf{b}_m$ as $\mathbf{x}_{\mathbf{b}_m}(t)$.
The Integral Curve Coordinate of a point $\mathbf{p}$ is
\begin{equation}
\label{eq:ICC_def}
    \mathcal{I}_{f_{D}}(\mathbf{p}) =
        \left\{
        \begin{array}{ll}
            ( \mathbf{b}_m, t ) & \text{if } \mathbf{p} \neq \mathbf{p}_0 \\
            \emptyset & \text{if } \mathbf{p} = \mathbf{p}_0
        \end{array}
        \right.
\end{equation}
where $\mathbf{x}_{\mathbf{b}_m}(t) = \mathbf{p}$
and $\mathbf{p}_0$ is the maximum point of $f_D$.

Since distinct integral curves never meet,
$\mathcal{I}$ is a bijection from points in the domain to Integral Curve Coordinates.
Given a bijective deformation $h$ of the boundary of $D$, $h(\partial D) = \partial D'$, we can similarly construct a special function
$f_{D'}$ on $D'$ to create $\mathcal{I}_{f_{D'}}$.
Transforming from $\mathcal{I}_{f_{D}}$ to $\mathcal{I}_{f_{D'}}$ can therefore be achieved by mapping the $b_m$ to $h(b_m)$.

\begin{figure}
\centering
\includegraphics[width=\linewidth]{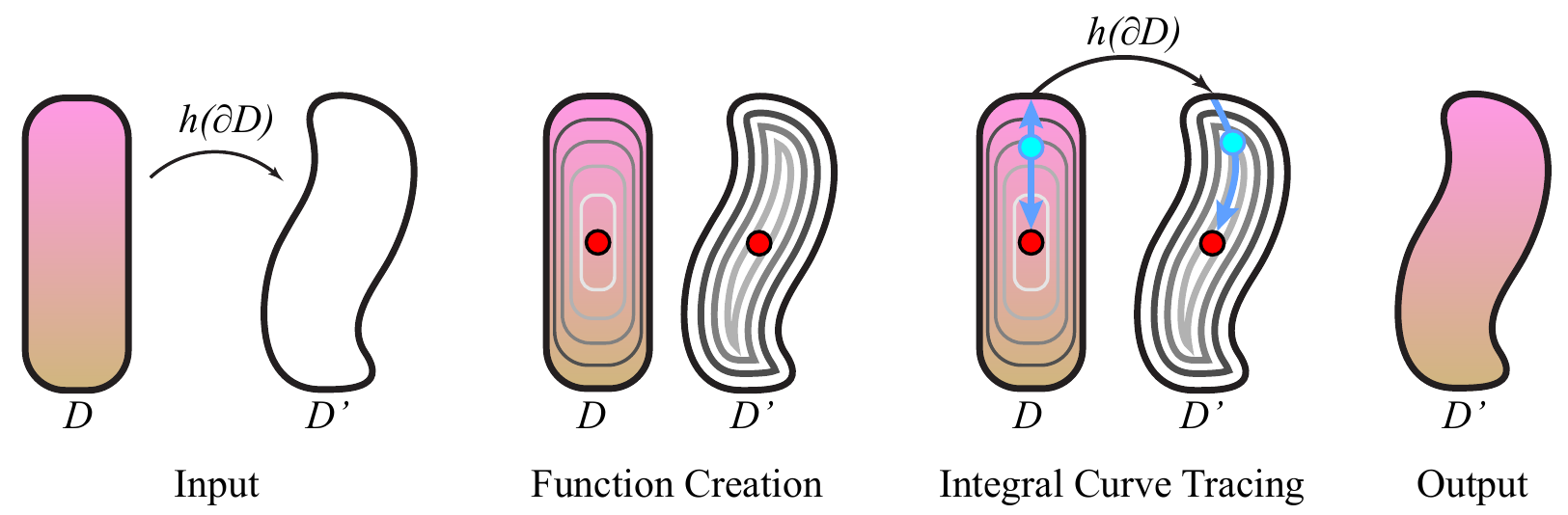}
\caption[Overview of our approach.]{An overview of our approach.
Given a shape $D$ and a homeomorphism of its boundary
$h\colon \partial D \rightarrow \partial D'$,
we first create two functions, one inside $\partial D$ and one inside $\partial D'$,
each with only a single critical point, a maximum.
To find the new position of a point $p \in D$,
we trace the integral curve passing through $p$ up to the maximum
and down to a boundary point $b_m$.
We then trace the integral curve from $h(b_m)$ on the boundary of $D'$ to the maximum.
The new position of $p$ is the point located the same fraction along the integral
curve in $D'$.
Repeating this process everywhere allows us to extend the boundary homeomorphism $h$ to the interior of $\partial D'$.
}
\label{fig:overview}
\end{figure}

To summarize, the steps to extend $h$ to a point $\mathbf{p}$ in the interior of $D$ are as follows (Figure~\ref{fig:overview}):
\begin{enumerate}
\item \label{enum:specialf} Create a function $f_D\colon D \rightarrow \mathbb{R}$ with a single critical point, a maximum, and whose gradients on the boundary point inward.
\item Create a similar function $f_D'$ for $D'$.
\item \label{enum:tracing} Compute the Integral Curve Coordinate $I_p = \mathcal{I}_{f_{D}}(\mathbf{p})$ by tracing the integral curve of $\nabla f_D$ in both directions from $\mathbf{p}$.
\item Transform $I_p$ to $I'_{p'}$ by mapping the boundary point $b_m$ of the Integral Curve Coordinate with $h$ (or, if $p=p_0$, identifying the unique maximum $p_0$ of $f_D$ with the unique maximum $p'_0$ of $f_D'$).
\item Compute $\mathcal{I}^{-1}_{f_{D'}}(I'_{p'})$ by tracing the integral curve of $\nabla f_{D'}$ from the boundary point to the maximum.
\end{enumerate}
(To warp points backwards from $D'$ to points in $D$, swap $D$ with $D'$ and $h$ with $h^{-1}$ in the above steps.)

Our implementation of this approach, described in the following sections,
creates functions on a regular grid discretization of the domain.
We trace integral curves on a piecewise linear interpolation of the function values.
In the piecewise linear setting, gradients are piecewise constant and discontinuous across piece boundaries.
This guarantees monotonic interpolation (no spurious critical points within an element)
and greatly reduces numerical issues in tracing integral lines.
However, it violates the assumption that integral lines never meet,
and can result in regions of the shape collapsing (though not inverting).
In $\mathbb{R}^2$, any of a number of $C^1$-continuous monotonic interpolation
techniques would eliminate this problem \cite{beatson1985monotonicity,floater1998tensor,Schmidt:1992:PMS,Carlson:1989:AMP}.

\section{Function Creation}

The first step in our deformation approach
is the creation of a suitable function
upon which to trace integral curves.
Namely, we require functions with a single critical point, a maximum, and whose gradients on the boundary point inward.
Our discrete implementation creates a function space satisfying the above requirements.
The function space takes the form of a set of inequality relationships on edges of a regular grid enclosing the shape boundary or cage.
We note that harmonic functions, which obey the strong maximum principle,
are suitable in 2D, but not in higher dimensions as they may contain spurious saddle points.

Our approach first chooses a grid vertex to be the maximum (Section~\ref{sec:grassfire})
and then
generates inequality relationships and vertex values such that all other points are regular (Section~\ref{sec:bfs}).
We then trace integral lines on a piecewise linear interpolation of the function values (Section~\ref{sec:integral_lines}).


\subsection{Maximum Selection}
\label{sec:grassfire}

Any grid vertex can be selected to be the maximum. However, because integral lines converge towards a maximum, we wish to choose a maximum vertex in the center of the shape. We find such a point with the grassfire transform~\cite{Blum:1967:TEN}
%
To do so, we construct a regular grid covering the mesh.
The grassfire transform iteratively ``burns away'' the boundary of a shape;
we use it to find the farthest point from the boundary.
Any vertex in the final iteration can be chosen arbitrarily.

\subsection{Assigning Function Values}
\label{sec:bfs}

Once the maximum is chosen, we assign function values to all grid vertices inside
and just enclosing the shape boundary or cage.
To do so, we build a function space containing functions that have a single critical point, a maximum, and gradients on the boundary that point inward.
The function space takes the form of a \emph{cousin tree} of inequality relationships
on edges of the regular grid.


\begin{definition}
Let $T$ be a tree defined on a subset of a regular grid.
We call $T$ a \emph{cousin tree} if
every pair of axis-aligned neighbor vertices in the tree 
are related to each other as parent-child or as tree cousins
(Figure~\ref{fig:tree_relationships}).
Two vertices are tree cousins if they are neighbors in the grid and their tree parents are also neighbors in the grid.
\end{definition}

\begin{figure}[h]
  \centering
  \includegraphics[width=0.5\linewidth]{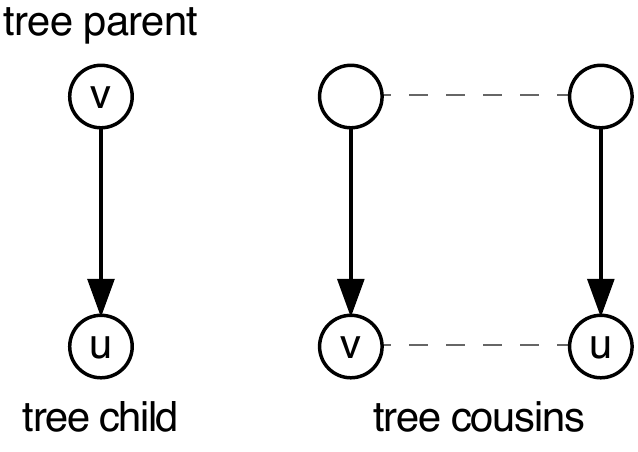}
  \caption[Cousin tree vertex relationships.]{Cousin tree vertex relationships: Two vertices $u,v$ that are axis-aligned neighbors in the grid must be either parent-child (left) or tree cousins (right). Tree cousins' parent vertices must also be axis-aligned neighbors in the grid.}
  \label{fig:tree_relationships}
\end{figure}

%

Our algorithm for creating a cousin tree,
illustrated in Figure~\ref{fig:cousin_tree_illustration},
is a form of breadth-first search (BFS) that,
for all vertices on the frontier in a given iteration,
always expands first along the first coordinate axis,
second along the second coordinate axis, and so on.
(The order of coordinate axes is not important so long as it is consistent across all iterations.) Pseudocode is as follows:
{\small
\begin{algorithmic}[1]
    \Procedure{CreateCousinTree}{$\mathit{maximumVertex}$}
    \State {$\rhd$ Store the tree as a set of directed edges:}
    \State {$\mathit{edges} \leftarrow \{\}$}
    \State {$\mathit{frontier} \leftarrow \{ \mathit{maximumVertex} \}$}
    \State {$\mathit{functionValue} \leftarrow$ empty dictionary}
    \State {$\mathit{functionValue}[ \mathit{maximumVertex} ] \leftarrow 1$ }
    \While {$\mathit{frontier}$ is not empty}
        \State{$\mathit{frontierNext} \leftarrow \{\}$}
        \For {$\mathit{dimension}$ in fixed dimension order}
            \For {$v$ in $\mathit{frontier}$}
                
                \State {$\mathit{candidates} \leftarrow \{$ unvisited neighbors of $v$ along $\mathit{dimension}$ direction $\}$}
                \State {$\mathit{edges} \leftarrow \mathit{edges} \cup {\{(v, n) \mid n \in \mathit{candidates}\}}$}
                
                \State {$\mathit{filtered} \leftarrow \{ n \in \mathit{candidates}$ if $n$ within boundary$\}$}
                
                \State {$\mathit{frontierNext} \leftarrow \mathit{frontierNext} \cup \mathit{filtered}$}
                
                \State {$\rhd$ Optional: directly assign function values:}
                \State {$\mathit{functionValue}[ \mathit{filtered} ] \leftarrow \mathit{functionValue}[ v ] - 1$}
                
            \EndFor 
        \EndFor 
        \State{$\mathit{frontier} \leftarrow \mathit{frontierNext}$}
    \EndWhile 
  \State {return $\mathit{edges}$, $\mathit{functionValue}$}
  \EndProcedure
\end{algorithmic}
}

When a vertex expands, it becomes the tree parent of the vertices it expands into.
In 2D, this means that all potential children along the $x$-axis are connected to the tree in one round. The next round, all potential children along the $y$-axis are connected to the tree, if they have not already been included.
While vertices outside the boundary are not added to the frontier
of the breadth-first search,
they are added to the tree and considered children of all adjacent vertices within the domain;
this ensures that vertices outside the boundary have smaller values than vertices inside,
which creates inward-pointing gradients at the boundary.
See the supplemental materials for a proof that this algorithm indeed creates a cousin tree.

\begin{figure}[h]
  \centering
  \includegraphics[width=1.0\linewidth]{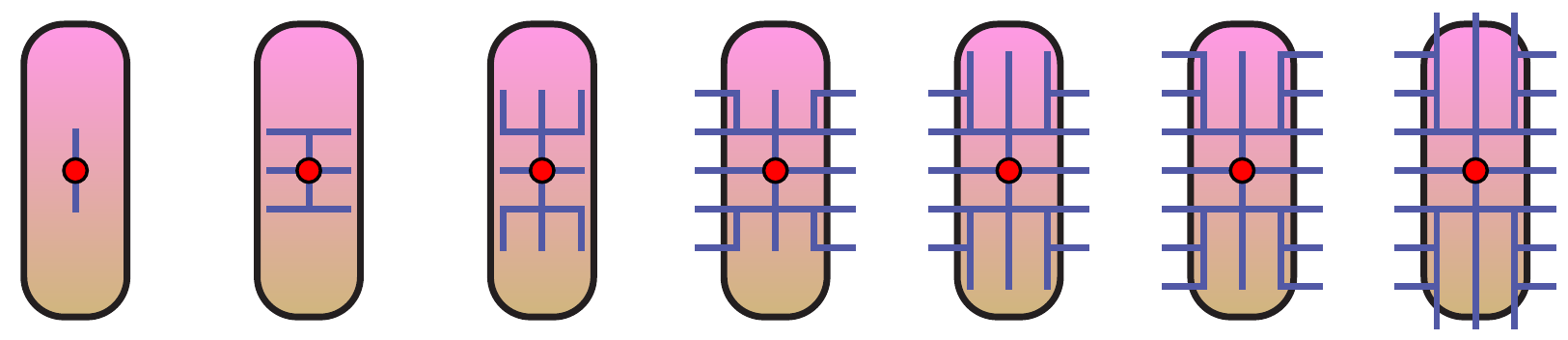}
  \caption[Cousin tree creation.]{A cousin tree is created via breadth-first search
  expanding outwards from a maximum (red circle).}
  \label{fig:cousin_tree_illustration}
\end{figure}

The cousin tree represents our function space;
the edge from a parent to a child vertex represents the inequality $f(parent) > f(child)$.
Note that the Manhattan distance lies within this function space,
as breadth-first search can be used to directly compute the graph distance to a given
node.
See the appendix for a proof that this function space contains no critical points
other than the maximum.

One the cousin tree is created,
we can assign values to grid vertices such that the tree parent always
has greater value than the tree child.
We have experimented with various approaches to assigning function values:
directly assigning the $L_1$ (tree) distance,
and solutions to the Laplace and bi-Laplace equations subject to various
boundary conditions and the cousin tree inequality constraints.
These approaches are evaluated in Section~\ref{sec:evaluation}.



\section{Tracing Integral Lines}
\label{sec:integral_lines}

With a suitable function $f_D$ in hand, we are ready to trace integral lines
to convert a point $\mathbf{p}$ in the interior of the shape to its Integral Curve Coordinate $\mathcal{I}_{f_{D}}(\mathbf{p}) = ( \mathbf{b}_m, t )$,
where $\mathbf{b}_m$ is the point on the boundary reached
by tracing the integral line through $\mathbf{p}$ in the decreasing direction
(gradient descent), and $t$ is the fraction of arc-length along the integral line
from $\mathbf{b}_m$ to the maximum, reached by gradient ascent
from either $\mathbf{p}$ or $\mathbf{b}_m$.
(Note that given a piecewise linear boundary or cage, we store $\mathbf{b}_m$
in barycentric coordinates with respect to the boundary piece it belongs to.)

Since our function values are provided on a regular grid, we require monotonic
interpolation of the values within grid cells,
so that spurious critical points are not introduced by the interpolation.
For our 2D implementation, we perform piecewise linear interpolation via a regular
triangulation that divides each grid cell into two triangles.
Notably, piecewise linear interpolation is monotonic and has constant gradient,
which mitigates problems resulting from numerical accuracy when tracing integral
lines and simplifies arc-length parameterization.

Because the gradient is constant within each triangle, one could trace the integral line
from a given starting point by simply computing the intersection of the ray
from the starting point in the gradient direction with the boundary of the triangle.
However, naive implementation of this numerical approach may create
discrepancies and inconsistencies such as loops
due to limitations in floating point precision.
In order to obviate numerical precision issues, we separate the topological calculation
(which edge of the triangle the integral line passes through)
from the geometric calculation (the coordinates of the point on the edge).
Our topological calculations are based on
comparing values and computing the signs of cross products,
which are amenable to symbolic perturbation
schemes~\cite{edelsbrunner1990simulation,Yap:GCT:1990,Yap:STG:1990}
for robust, exact evaluation.\footnote{We did not implement such a scheme,
as we did not encounter numerical issues with
our floating point implementation.}

Our algorithm traces an integral line by iteratively computing
the sequence of triangle edges it intersects (and points on those edges).
As a special, initial case, the integral line may originate from a point inside a triangle,
but thereafter will be tracked via the triangle edges it intersects.
Our algorithm proceeds as follows,
and is illustrated in Figure~\ref{fig:pathfronts}.
%
Given a point $p_1$ located on the edge $e_{B_3}$ of a triangle $B$, we determine the next edge
by computing the sign of the cross product between (a) the gradient of the function
within $B$ and (b) the vector from $p_1$ towards the vertex of the triangle
opposite the edge $e_{B_3}$ (the dotted red line in Figure~\ref{fig:pathfronts}, right).
The sign of the cross product of these two vectors determines the outgoing edge of the integral line
(Table~\ref{table:edge_sign} with edge labeling given by
Figure~\ref{fig:regular_triangles}).
This robustly and stably computes the outgoing edge.
The point on the outgoing edge ($p_2$ in Figure~\ref{fig:pathfronts}, right)
is then determined via standard line/line intersection computation
(which is made simple since triangles edges are aligned with grid edges or diagonal).
Tracing proceeds in the triangle on the other side of the outgoing edge.

In the special, initial case of an integral line originating at a point $p_0$ 
inside a triangle $A$, the sign of the cross product of the gradient is computed with each of
the three vectors from the point to the triangle's vertices
(Figure~\ref{fig:pathfronts}, left).
The sign of the cross products determines which vectors
the gradient is to the left and right of, which indicates the outgoing edge accordingly.
The point on the edge can then be computed numerically, and the general case of the
algorithm proceeds.

\begin{figure}[h]
  \centering

	\includegraphics[width=.4\linewidth]{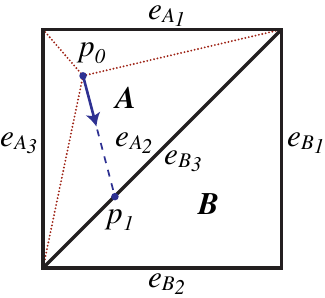}
	\hfill
	\includegraphics[width=.4\linewidth]{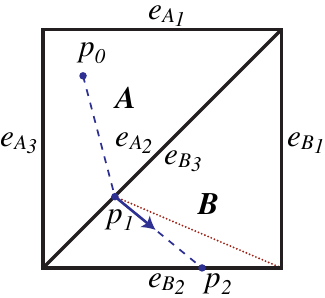}

  \caption[Fetching the next location.]{An integral line is traced (left) from $p_0$ inside triangle $\textbf{A}$.
  The gradient direction (the purple arrow) is compared to the
  red dotted lines, indicating which edge of $\textbf{A}$
  is intersected. Line/line intersection tells us
  the numerical location of the next point $p_1$ lying on a triangle edge,
  which is the general case.
  Tracing then proceeds inside the opposite
  triangle $\textbf{B}$ (right), and the gradient direction
  only needs to be compared to one dotted red line.}
  \label{fig:pathfronts}
\end{figure}

\begin{figure}[h]
  \centering  
	\includegraphics[width=1.0\linewidth]{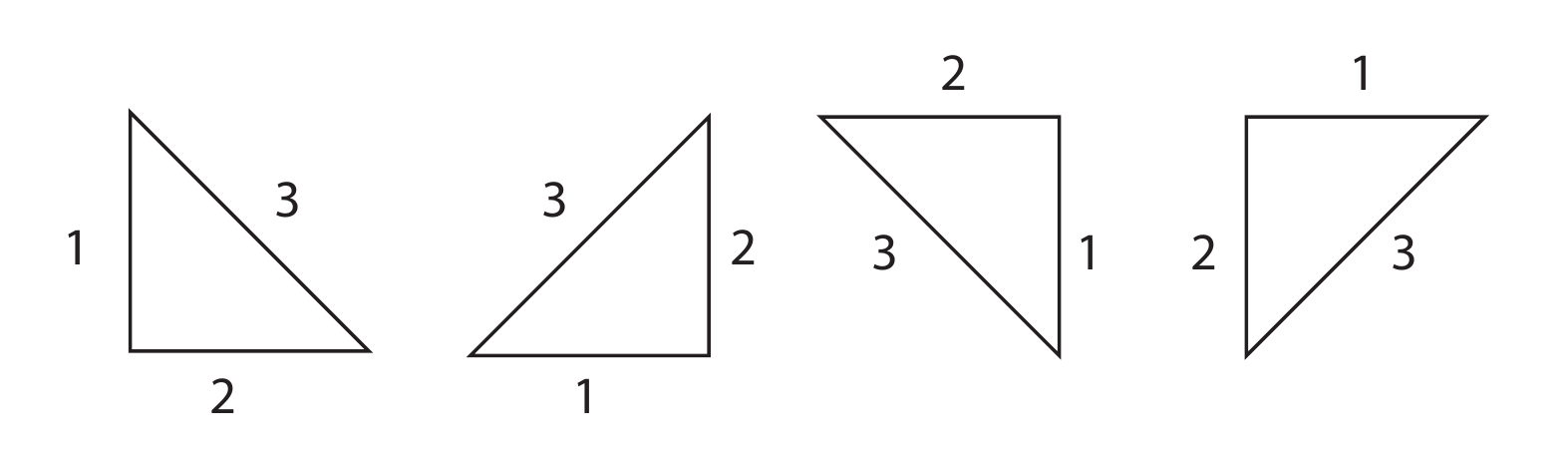}
  \caption[Base triangles.]{The edge labeling for triangles in the regular grid, such that the lookup table in Table~\ref{table:edge_sign} can be used to compute the outgoing edge when tracing integral curves.}
  \label{fig:regular_triangles}
\end{figure}

Finally, because our piecewise linear domain is not $C^1$, integral lines may meet.
This manifests in our integral line tracing algorithm as a triangle whose gradient points
backwards against the incoming edge. We call such edges \emph{compression edges}.
Without special care, the integral line would
zig-zag or staircase between the two triangles.
We detect this by testing whether the triangle's gradient points towards the
incoming edge (e.g. right back out of the triangle).
If so, then tracing
return to the face opposite the incoming edge,
and the next point on the integral line is the endpoint of the edge
pointed towards by the gradient.
Because our triangles triangulate a regular grid,
the dot product of the gradient with the edge itself is
a simple sign test or comparison between components of the gradient.
This procedure simulates the integral curve bouncing between
the two triangles as their gradients merge, skipping the intermediate steps.

\begin{table}
  \centering
  \begin{tabular}{ccc}
    \toprule
    Edge & Sign & Next Edge \\
    \midrule
    1 & + & 3 \\
    1 & - &  2 \\
    2 & + &  1 \\
    2 & - &  3 \\
    3 & + &  2 \\
    3 & - & 1 \\
    \bottomrule
  \end{tabular}

\caption{Our algorithm for tracing integral lines jumps from one triangle edge to the next (Figure~\ref{fig:pathfronts}, right).
With edges labeled according to Figure~\ref{fig:regular_triangles},
the next edge is determined by the sign of the cross product between the function gradient and the vector towards the triangle's opposite corner.}
\label{table:edge_sign}
\vspace{-1em}
\end{table}

\section{Evaluation}
\label{sec:results}
\label{sec:evaluation}

We experimented with various functions for Integral Curve Coordinates.
Figure~\ref{fig:functions} depicts the computed function, integral lines, and deformations
for the following functions subject to cousin tree constraints:
$L_1$ (tree) distance and
solutions to the Laplace and bi-Laplace equation with boundaries equal to zero;
and solutions to the Laplace equation with boundaries equal to zero,
without constraints imposed by the cousin tree.
%
We wish for integral curves to be maximally separated.
Solutions to the Laplace equation when boundary values are fixed to zero
result in integral curves that intersect the boundary orthogonally,
resulting in the most separation between integral lines.
There is little difference between solutions to the Laplace equation
with and without the cousin tree constraints.

\begin{figure}
\centering
\small

\begin{tabular}{@{}c@{\hspace{1mm}}c@{}c@{}c@{}c@{}}
undeformed & \parbox[b]{.19\linewidth}{$L_1$ (tree)\\distance} & \parbox[b]{.19\linewidth}{bi-Laplace\\equation\\with\\cousin tree\\constraints} & \parbox[b]{.19\linewidth}{Laplace\\equation\\with\\cousin tree\\constraints} & \parbox[b]{.19\linewidth}{Laplace\\equation\\without\\cousin tree\\constraints} \\
\midrule
\includegraphics[width=.19\linewidth]{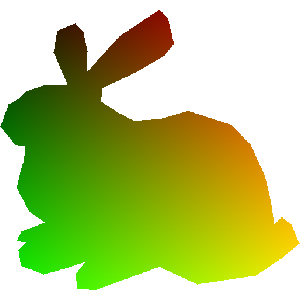}
&
\includegraphics[width=.19\linewidth]{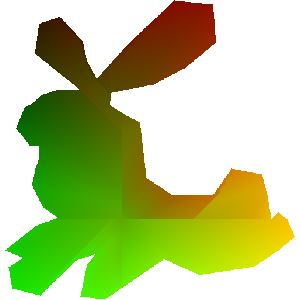}
&
\includegraphics[width=.19\linewidth]{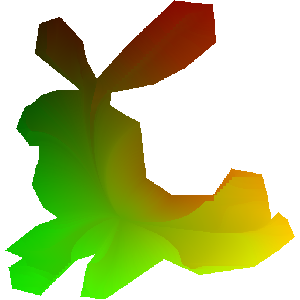}
&
\includegraphics[width=.19\linewidth]{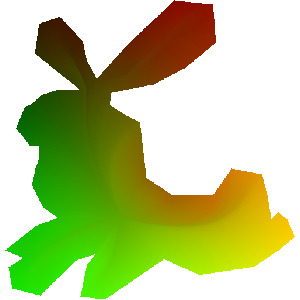}
&
\includegraphics[width=.19\linewidth]{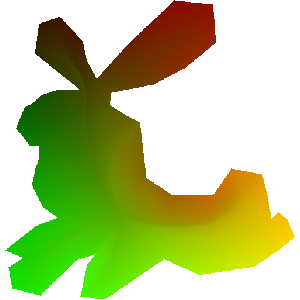}
\\
\includegraphics[width=.19\linewidth]{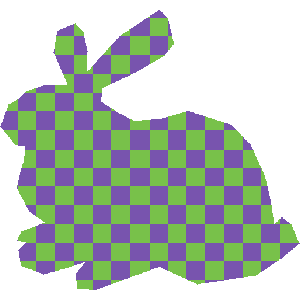}
&
\includegraphics[width=.19\linewidth]{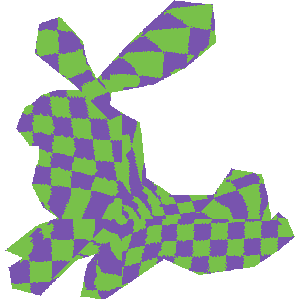}
&
\includegraphics[width=.19\linewidth]{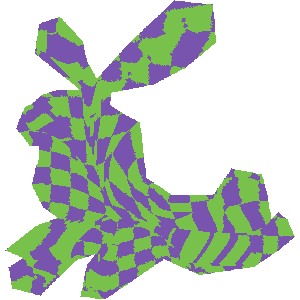}
&
\includegraphics[width=.19\linewidth]{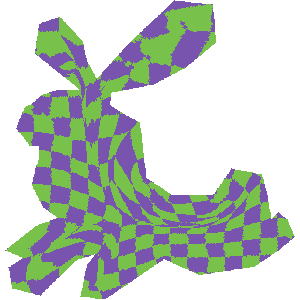}
&
\includegraphics[width=.19\linewidth]{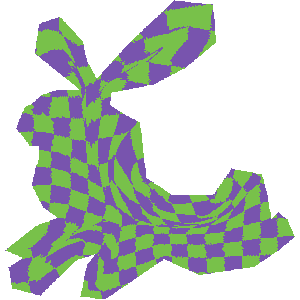}
\\
 &
\includegraphics[width=.19\linewidth]{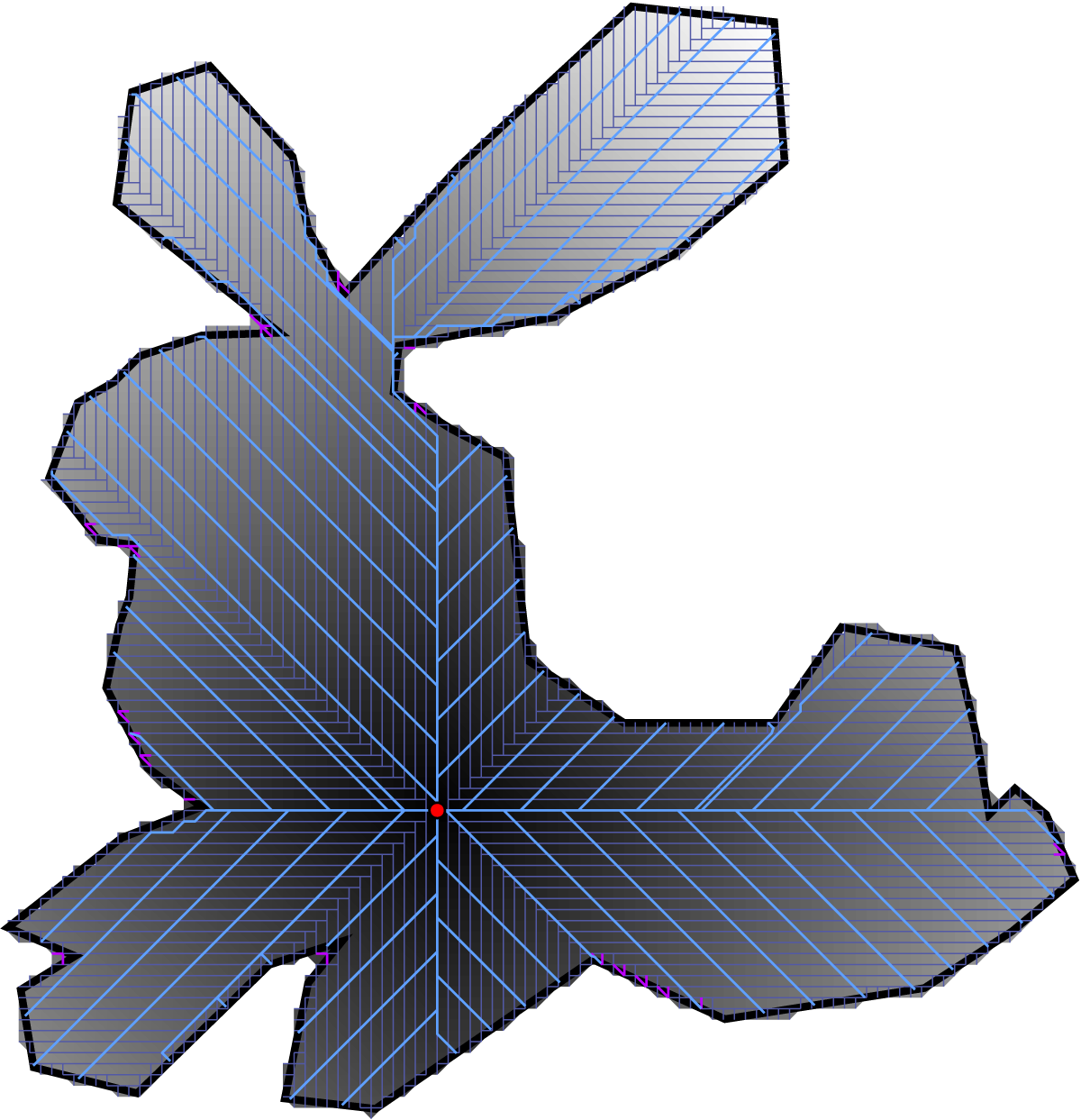}
&
\includegraphics[width=.19\linewidth]{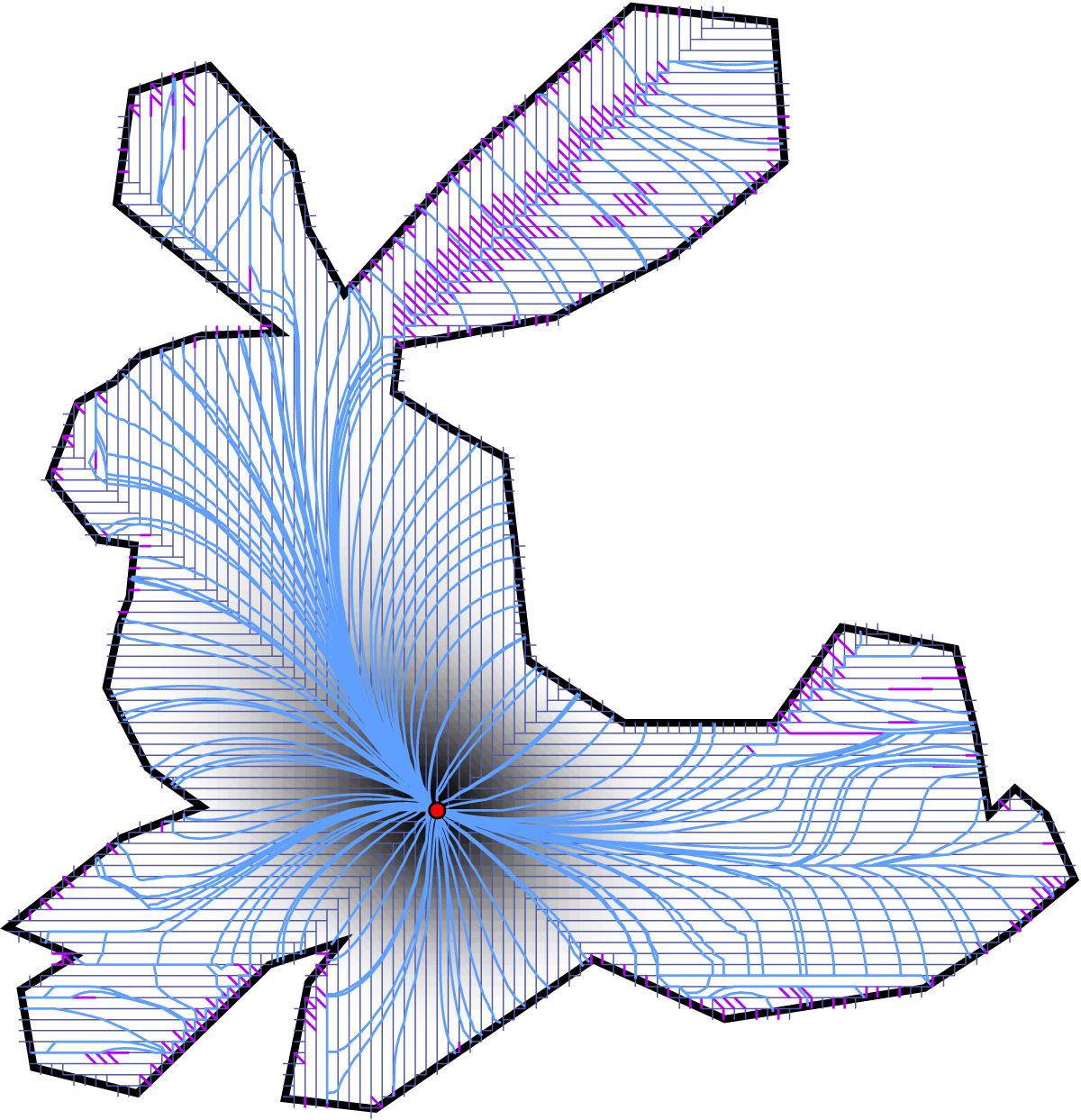}
&
\includegraphics[width=.19\linewidth]{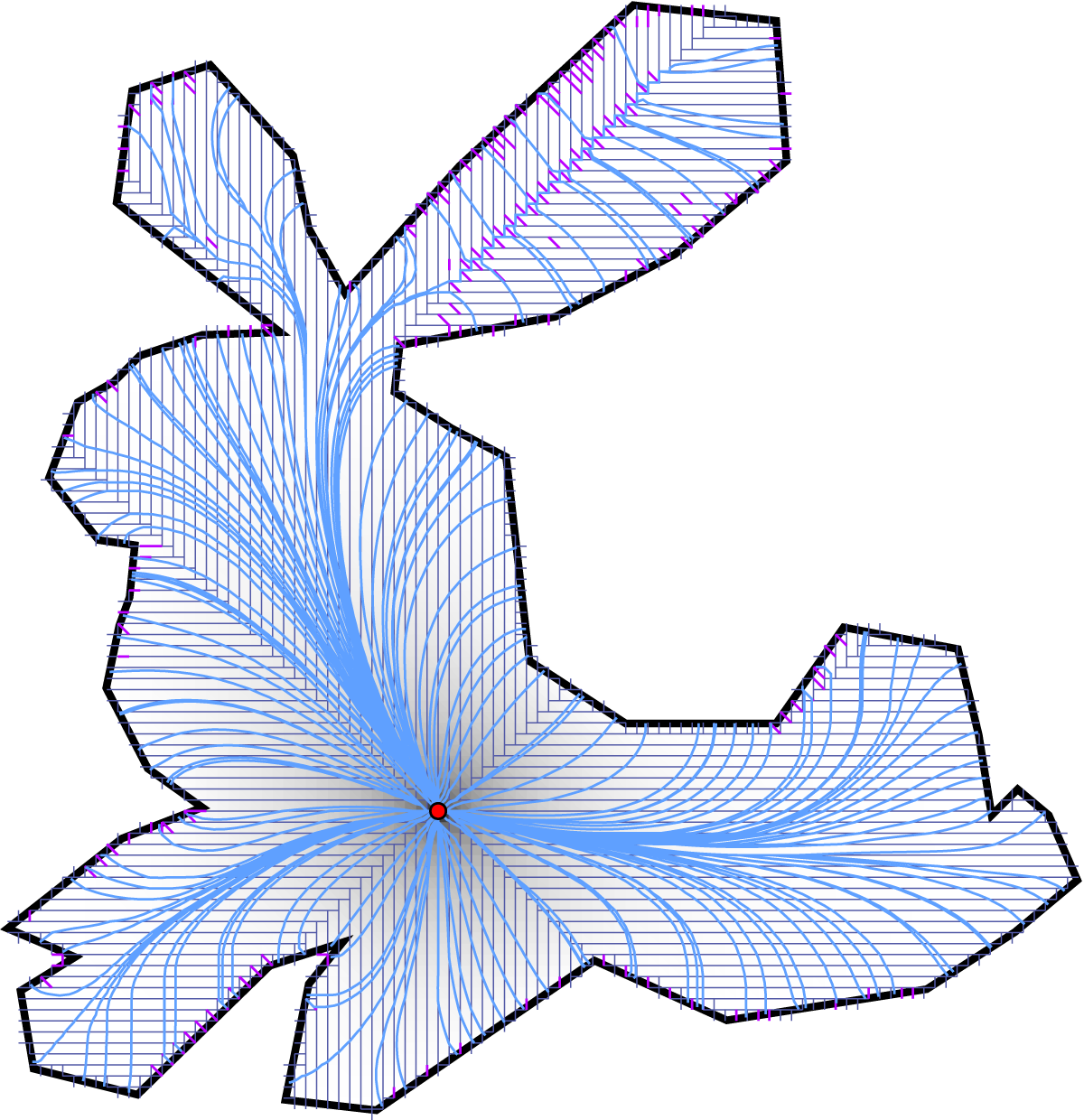}
&
\includegraphics[width=.19\linewidth]{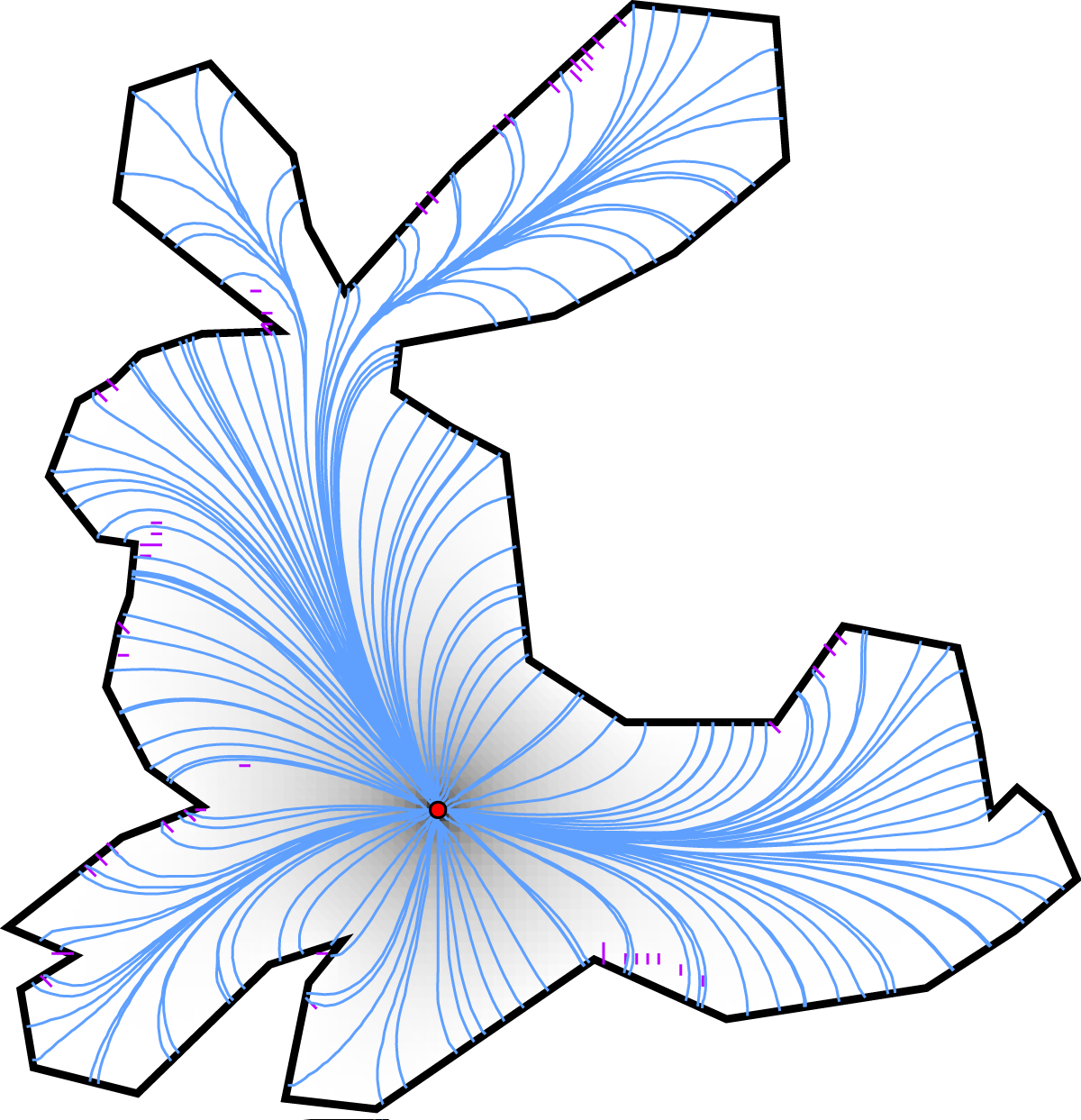}
\\
\includegraphics[width=.19\linewidth]{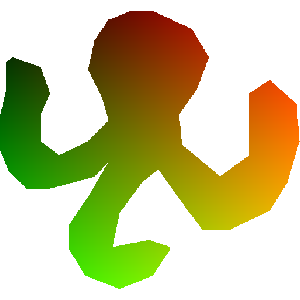}
&
\includegraphics[width=.19\linewidth]{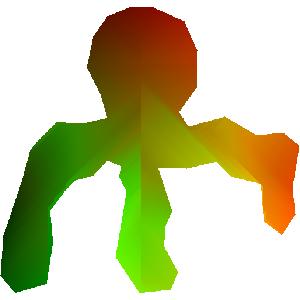}
&
\includegraphics[width=.19\linewidth]{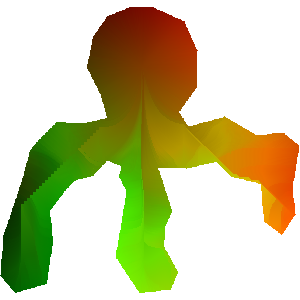}
&
\includegraphics[width=.19\linewidth]{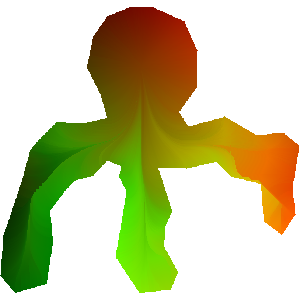}
&
\includegraphics[width=.19\linewidth]{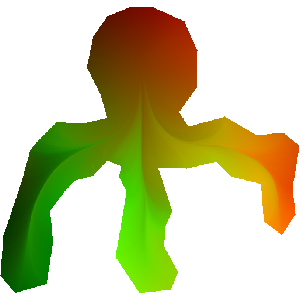}
\\
\includegraphics[width=.19\linewidth]{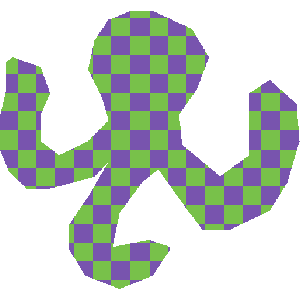}
&
\includegraphics[width=.19\linewidth]{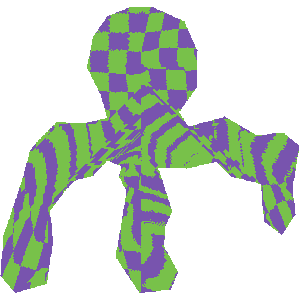}
&
\includegraphics[width=.19\linewidth]{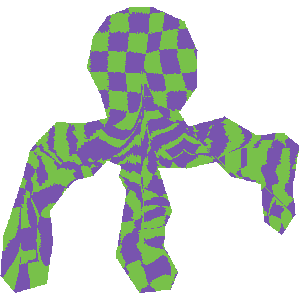}
&
\includegraphics[width=.19\linewidth]{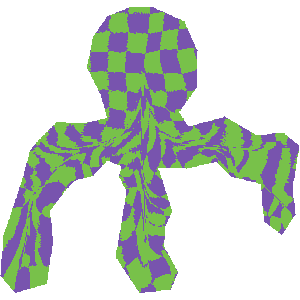}
&
\includegraphics[width=.19\linewidth]{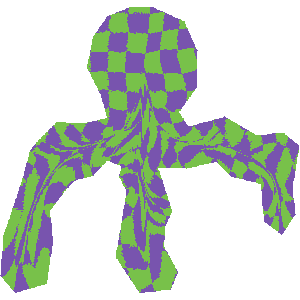}
\\
 &
\includegraphics[width=.19\linewidth]{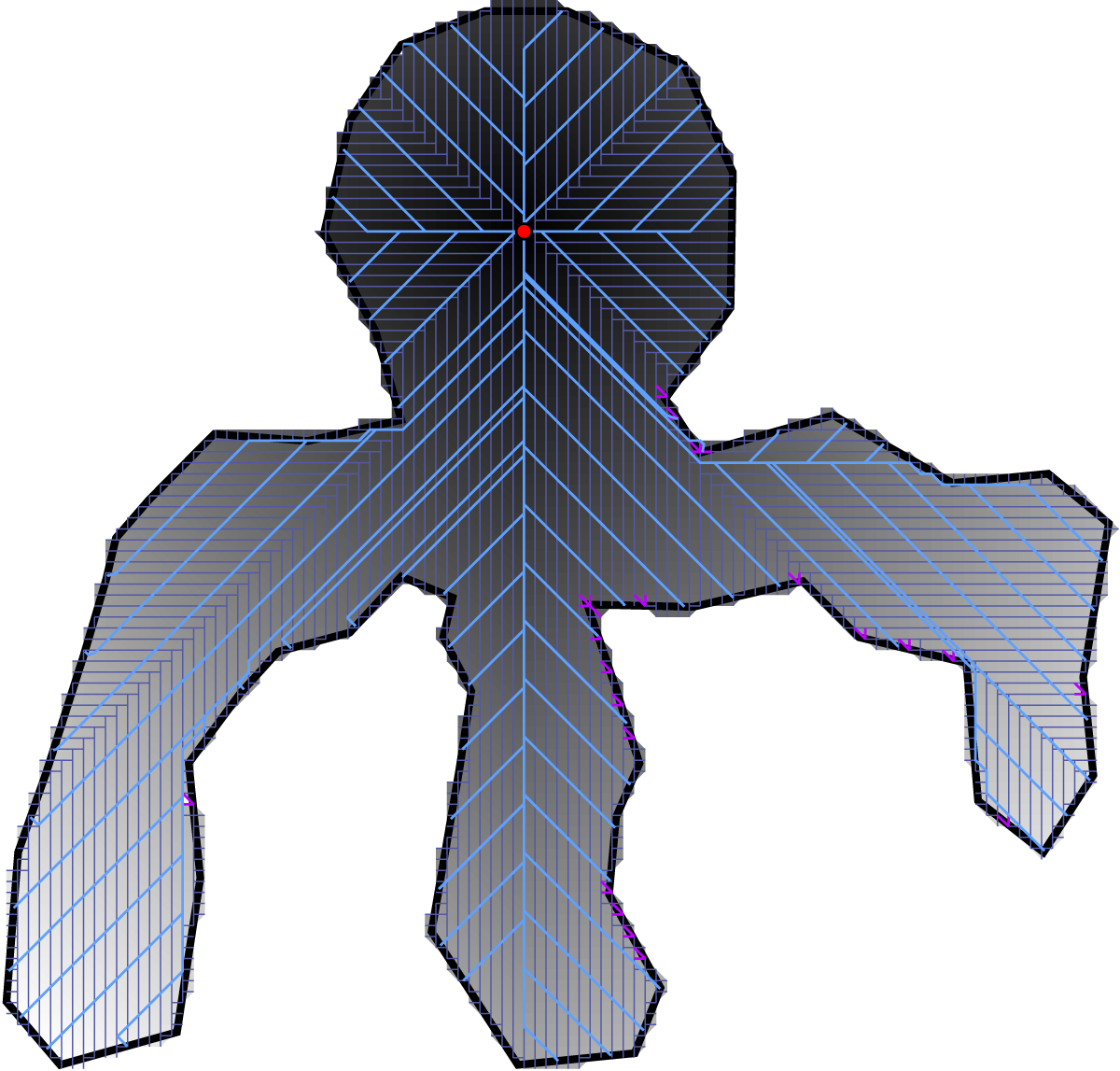}
&
\includegraphics[width=.19\linewidth]{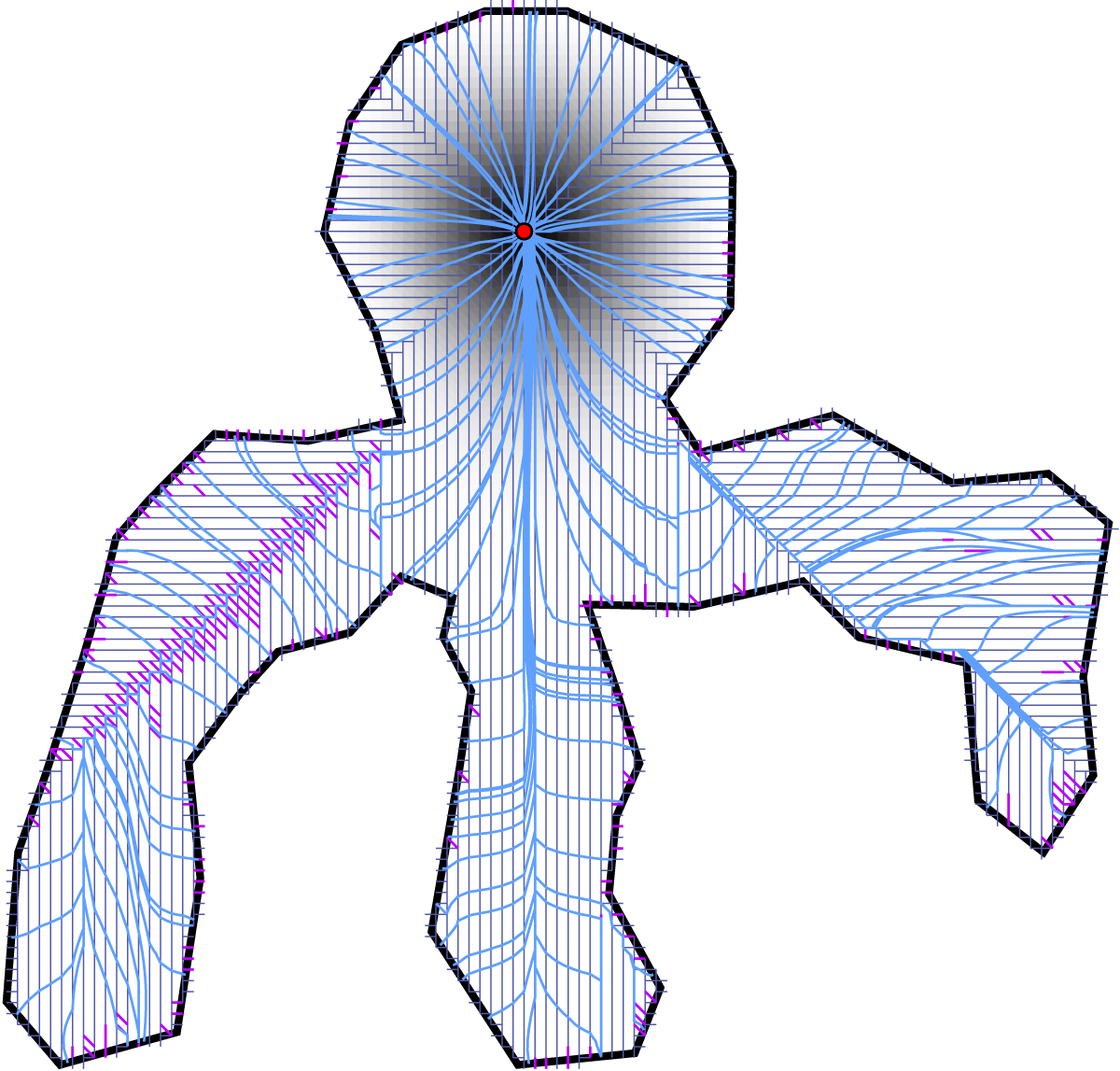}
&
\includegraphics[width=.19\linewidth]{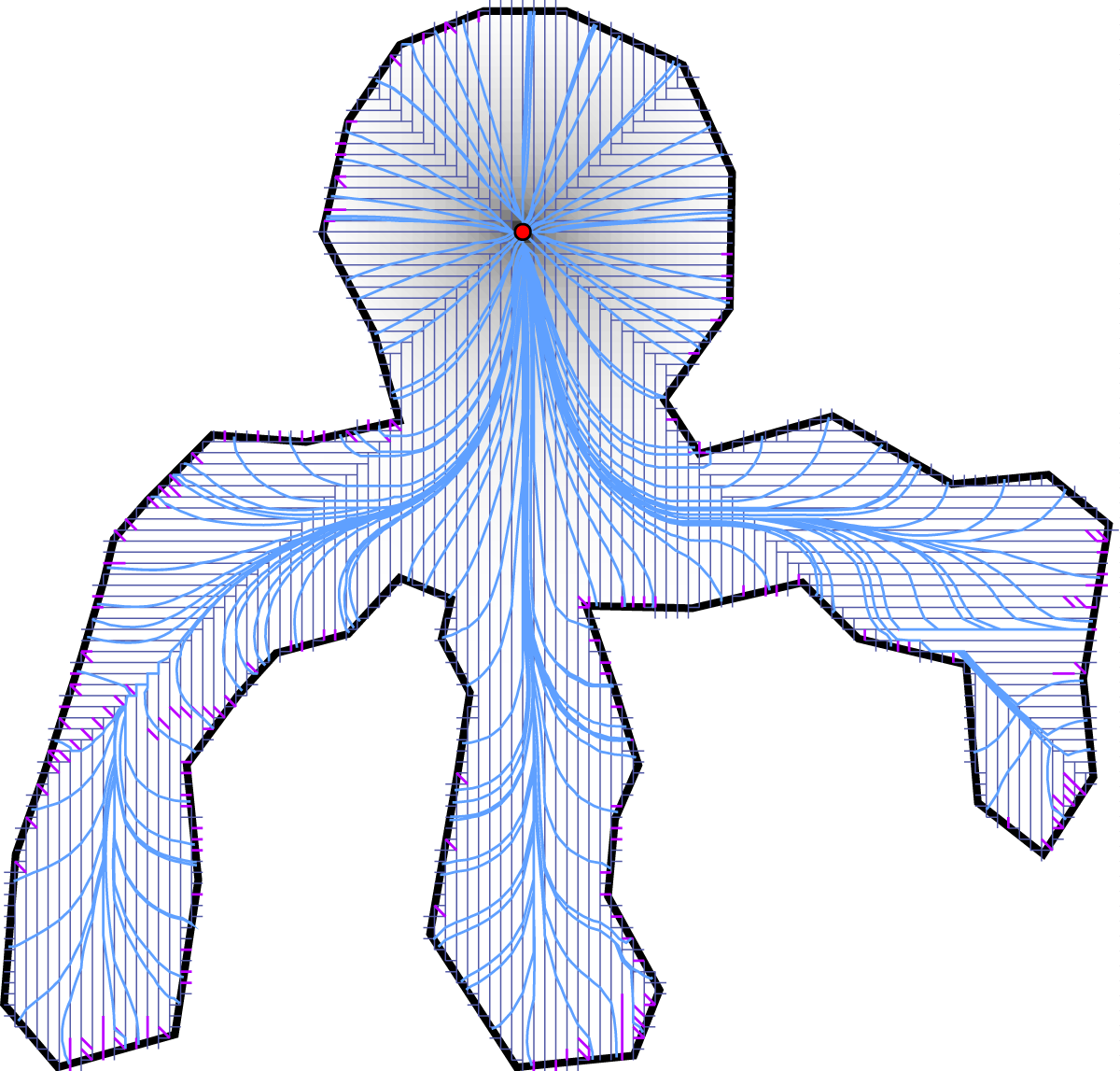}
&
\includegraphics[width=.19\linewidth]{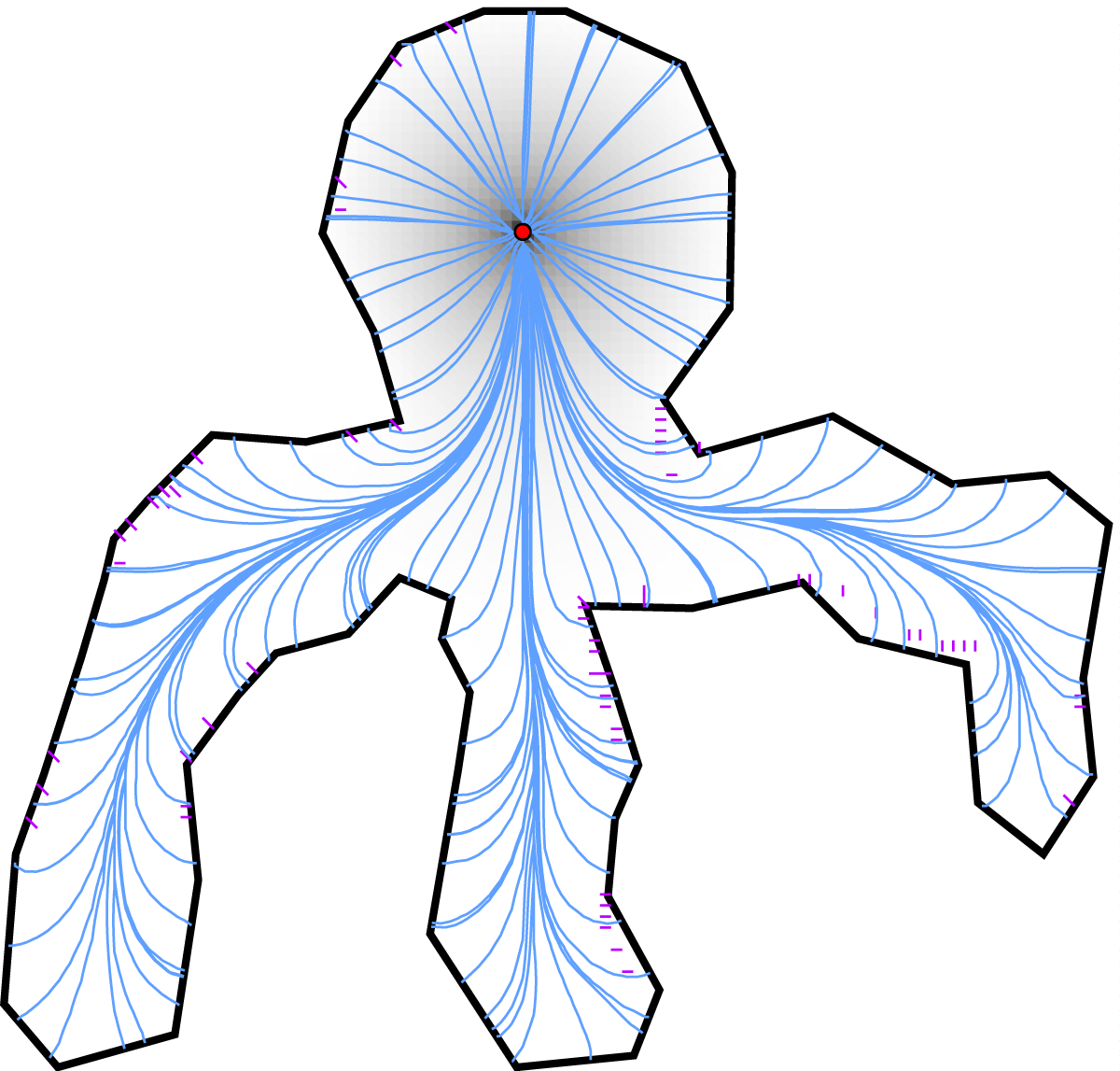}
\\
\includegraphics[width=.19\linewidth]{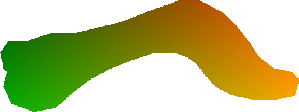}
&
\includegraphics[width=.19\linewidth]{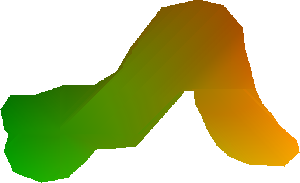}
&
\includegraphics[width=.19\linewidth]{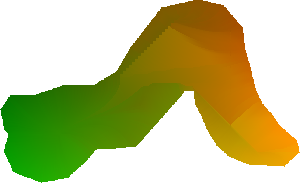}
&
\includegraphics[width=.19\linewidth]{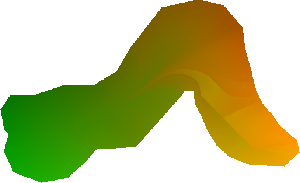}
&
\includegraphics[width=.19\linewidth]{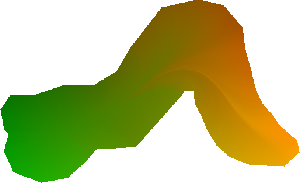}
\\
\includegraphics[width=.19\linewidth]{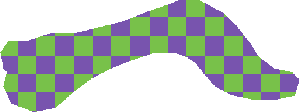}
&
\includegraphics[width=.19\linewidth]{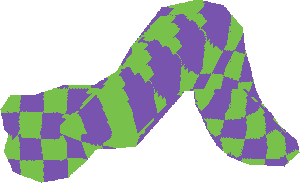}
&
\includegraphics[width=.19\linewidth]{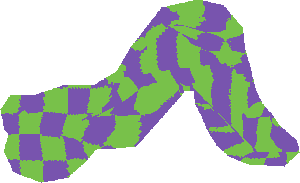}
&
\includegraphics[width=.19\linewidth]{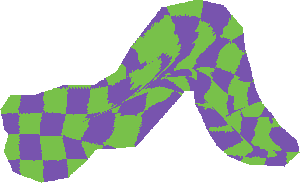}
&
\includegraphics[width=.19\linewidth]{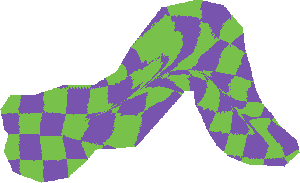}
\\
 &
\includegraphics[width=.19\linewidth]{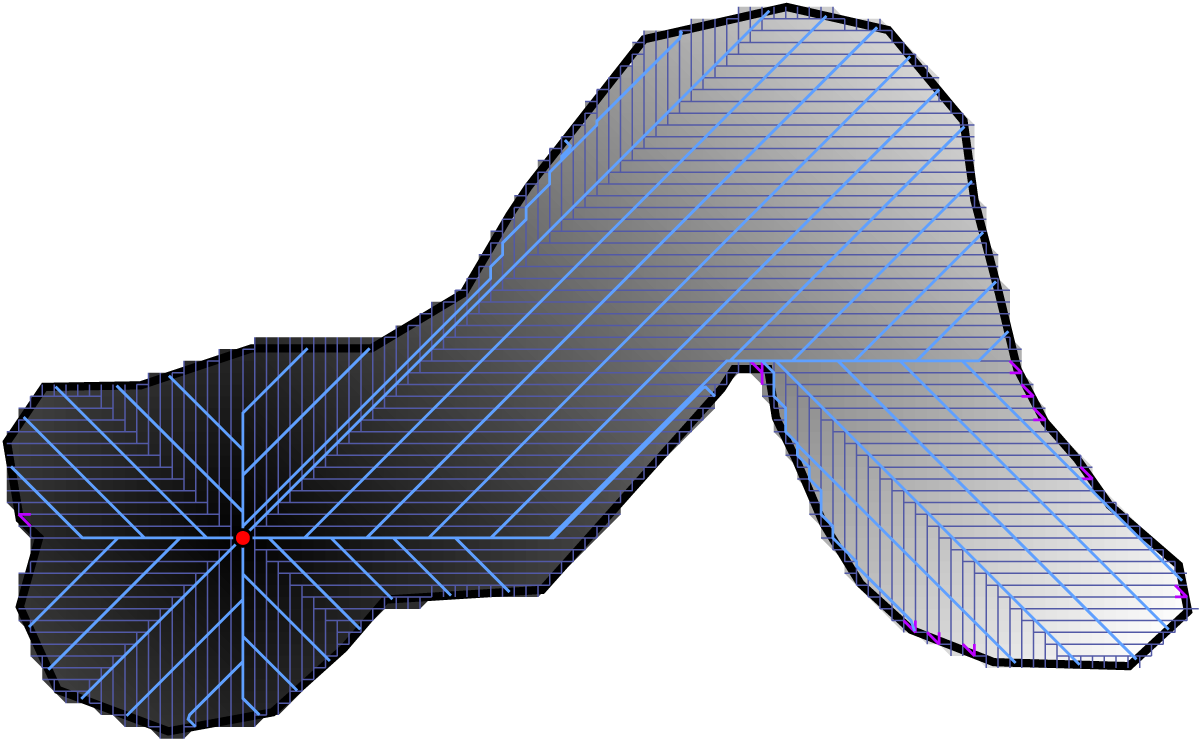}
&
\includegraphics[width=.19\linewidth]{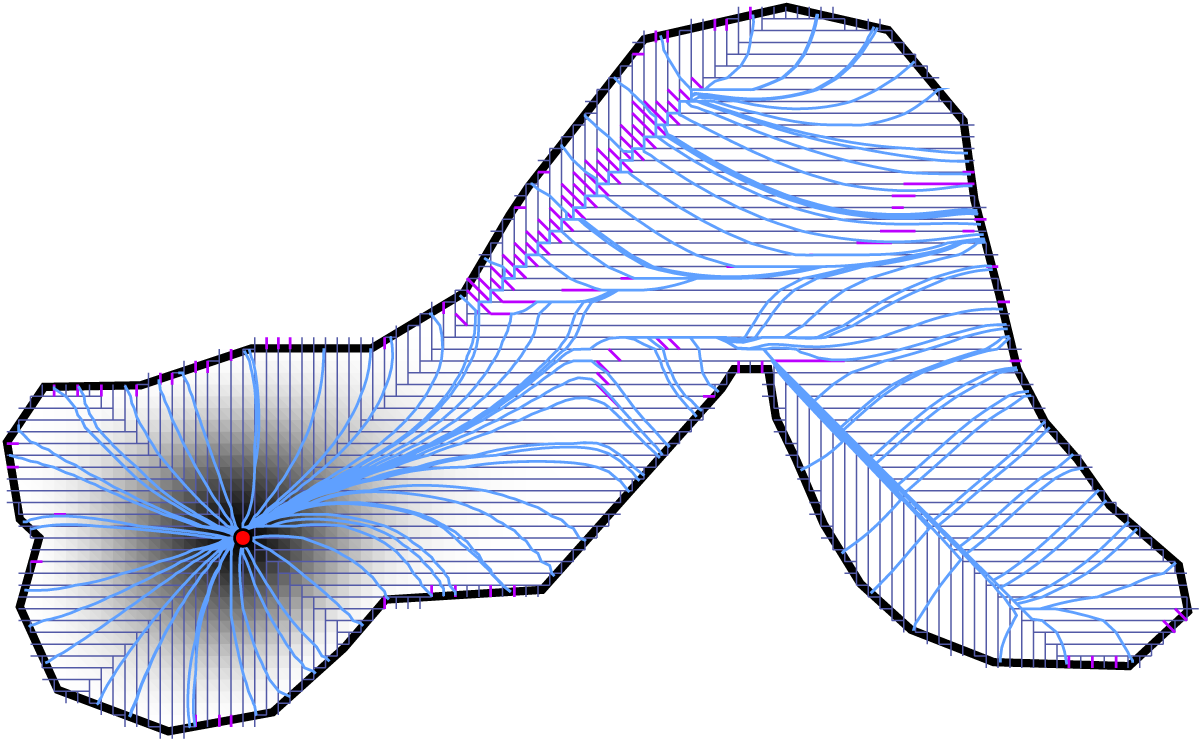}
&
\includegraphics[width=.19\linewidth]{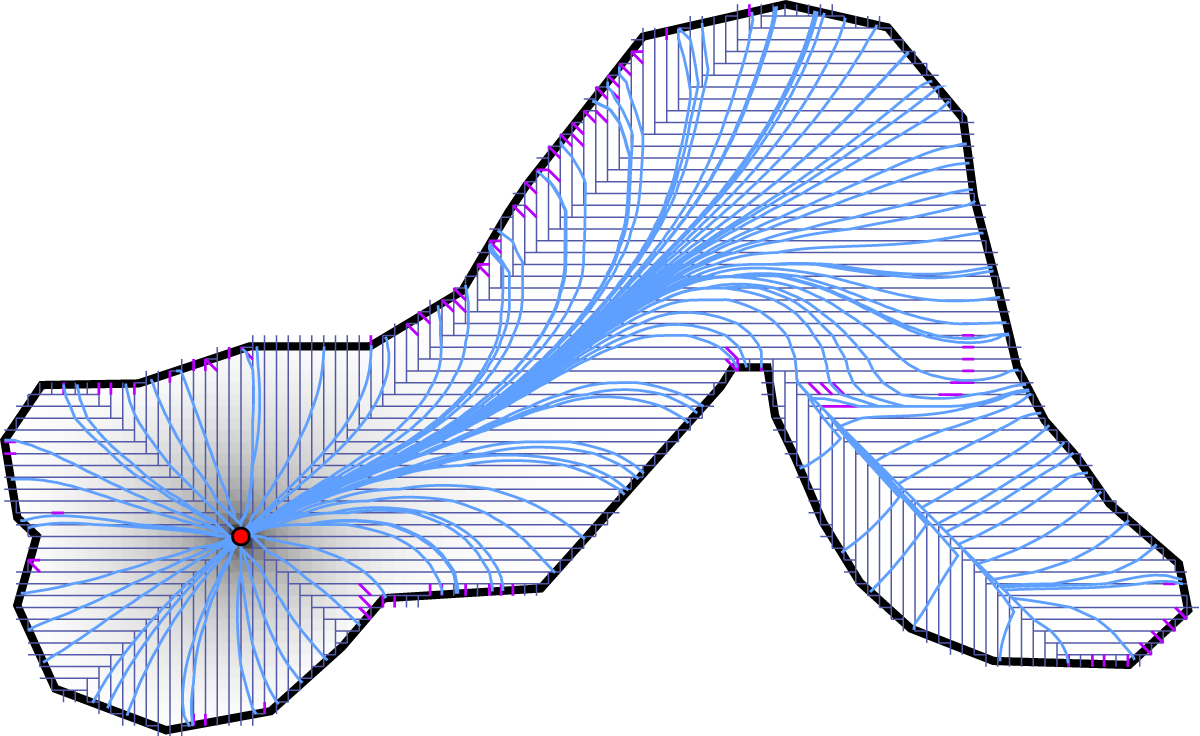}
&
\includegraphics[width=.19\linewidth]{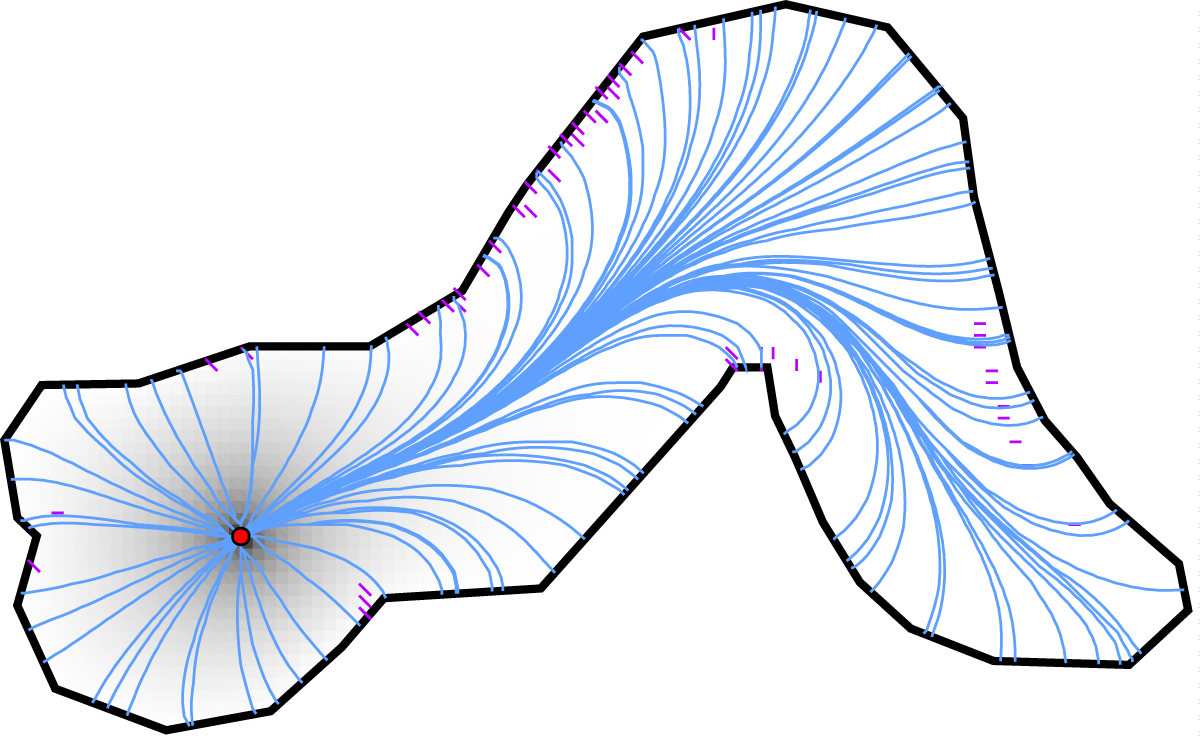}
\end{tabular}

\caption{Integral Curve Coordinates 
computed using various functions.
For each object, the first row depicts the $uv$ deformation map in the red and green channels;
the second row warps a checkerboard pattern;
and the third row displays the function values (background lightness),
maximum location (red circle)
integral lines (blue),
the cousin tree, and expansion edges (purple, Section~\ref{sec:expansion}).
}
\label{fig:functions}
\end{figure}

Figures~\ref{fig:compare_big}--\ref{fig:compare_composite} compare the results
of our deformation approach to Complex Barycentric Coordinates
\cite{weber2009complex}, Controllable Conformal Maps \cite{weber2010controllable}, 
Composite Mean Value Mappings \cite{schneider2013bijective},
and Locally Injective Mappings \cite{schuller2013locally}.
Our deformation, while less fair,
can be generalized to any dimension.
Although some integral lines do flatten out as a result of deformation,
no inverted elements are created by our approach.
Our deformation can be used to compute a correct starting configuration for
recent techniques which improve the fairness of deformations while preserving
bijectivity~\cite{schuller2013locally,aigerman2013injective}.

The long and narrow shape in Figure~\ref{fig:compare_composite} emphasizes the degree
to which the location of the maximum affects the overall deformation.
We believe that replacing the single maximum point by a skeleton or medial axis
is a fruitful direction for future research.

\begin{figure}[h]
	\centering
	\resizebox{\columnwidth}{!}{
	\begin{tabular}{rrr}
	  \includegraphics[width=.3\linewidth]{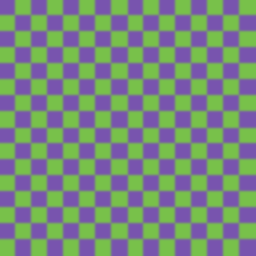}
	  &
	  \includegraphics[width=.3\linewidth]{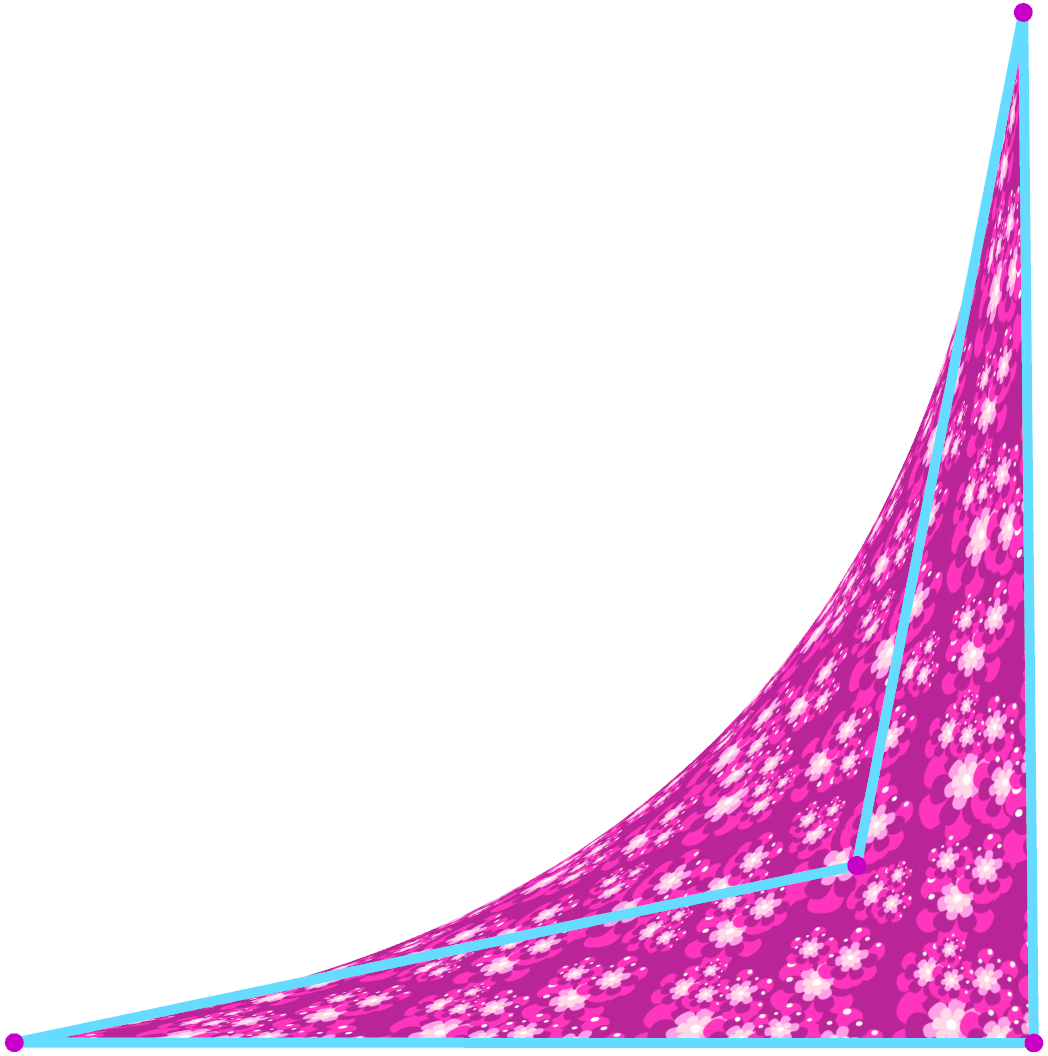}
	  &
	  \includegraphics[width=.3\linewidth]{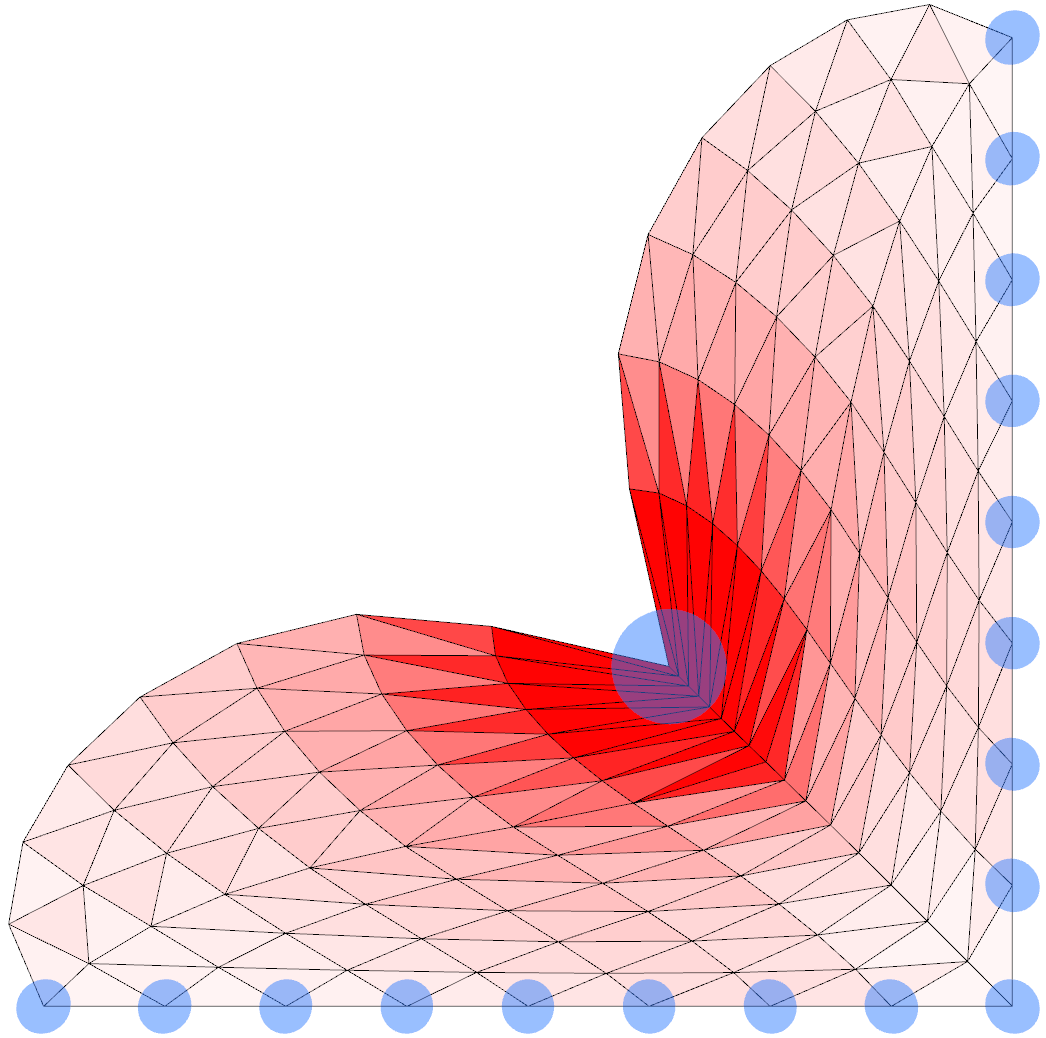}
	  \\
	  \includegraphics[width=.3\linewidth]{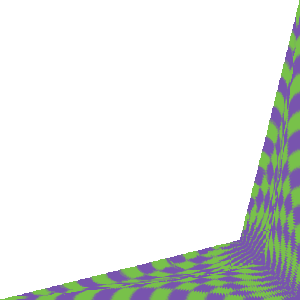}
	  &
	  \includegraphics[width=.3\linewidth]{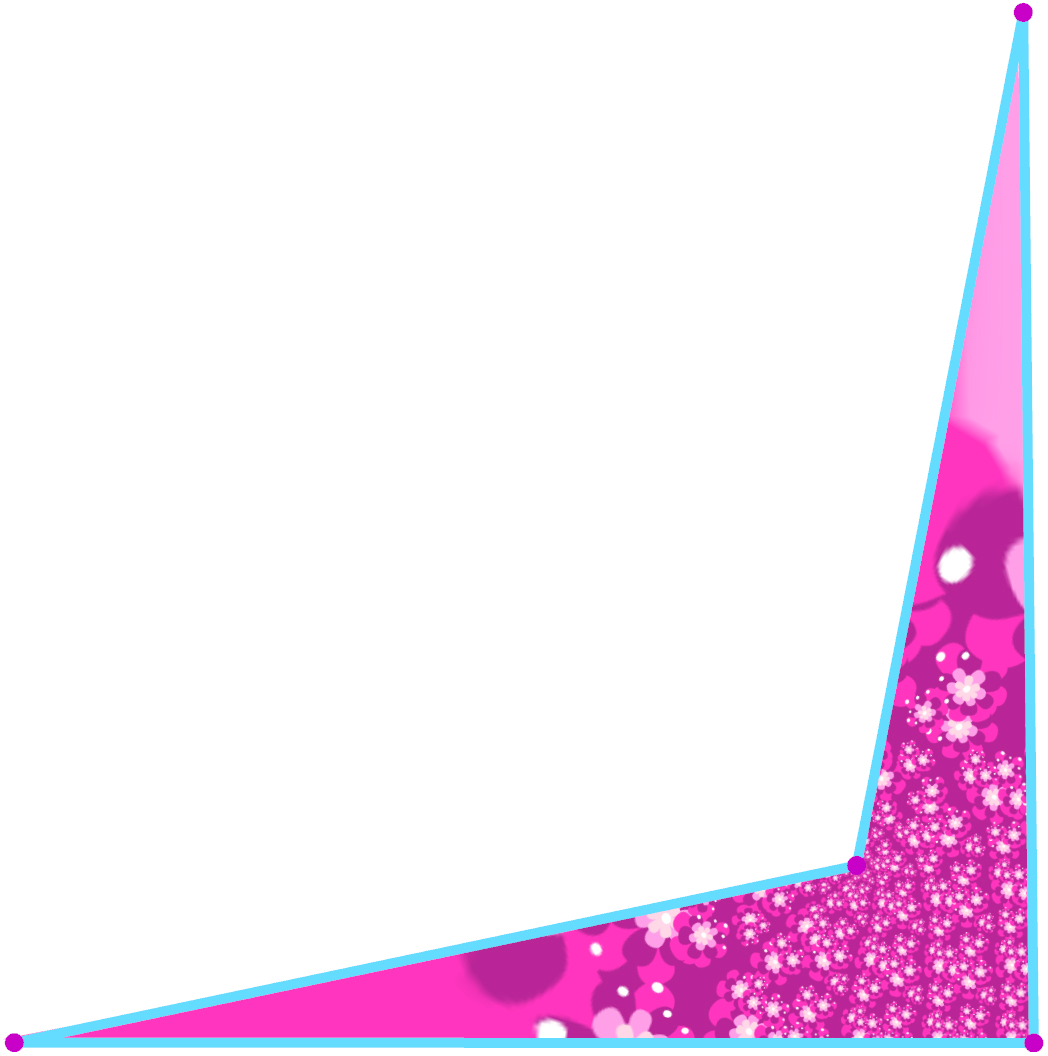}
	  &
	  \includegraphics[width=.3\linewidth]{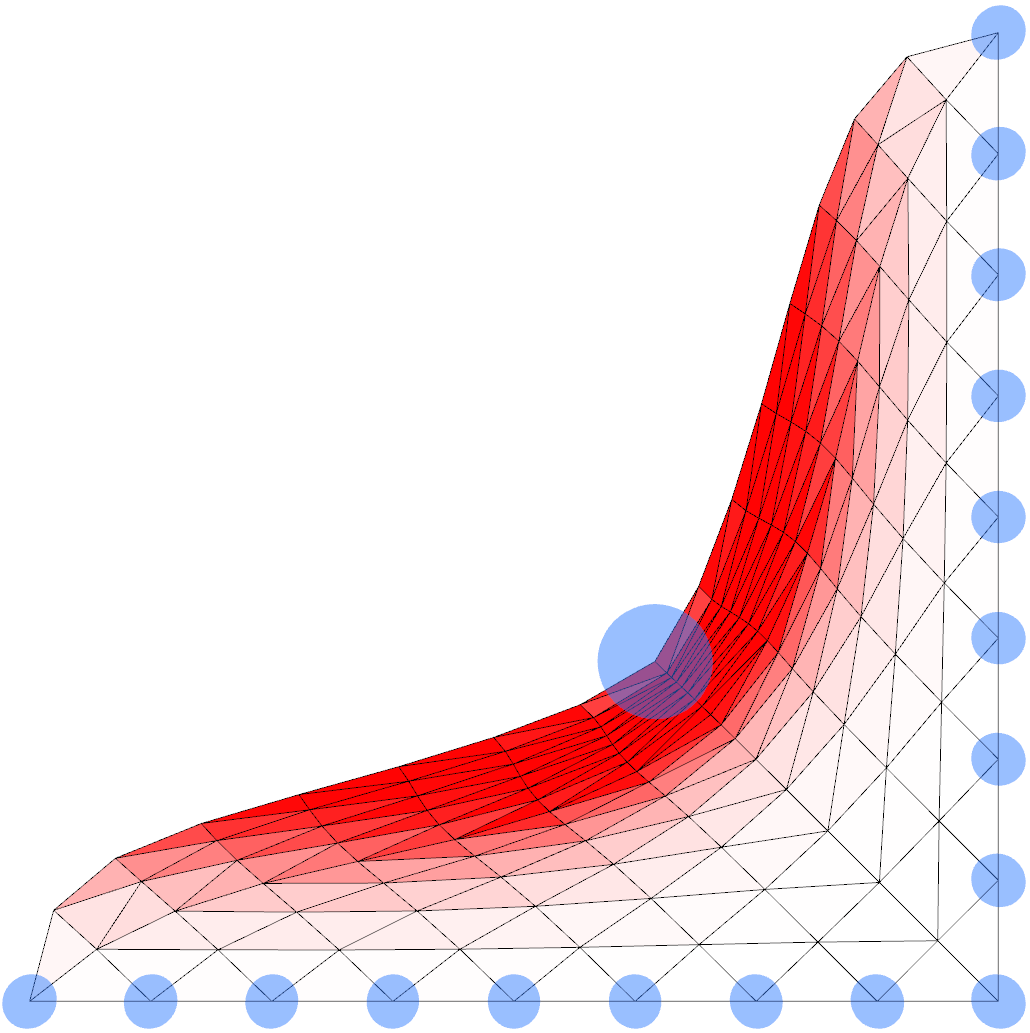}
	\end{tabular}
	}
	\caption{An extreme deformation of a square, a counter-example to the bijectivity of generalized barycentric systems~\protect\cite{Jacobson:2013:BMG}.
	Left column: A square shape with a checkerboard pattern, Integral Curve Coordinates with a Laplace equation not subject to cousin tree constraints.
	Middle column: Harmonic coordinates~\protect\cite{Joshi:2007:HCC}, Controllable Conformal Maps~\protect\cite{weber2010controllable}.
	Right column: Locally Injective Mapping using Dirichlet and Laplacian energy~\protect\cite{schuller2013locally}.
	Only Integral Curve Coordinates and Controllable Conformal Maps generate a bijective deformation.}
	\label{fig:compare_big}
\end{figure}


\begin{figure}[h]
	\centering
	  \includegraphics[width=.3\linewidth]{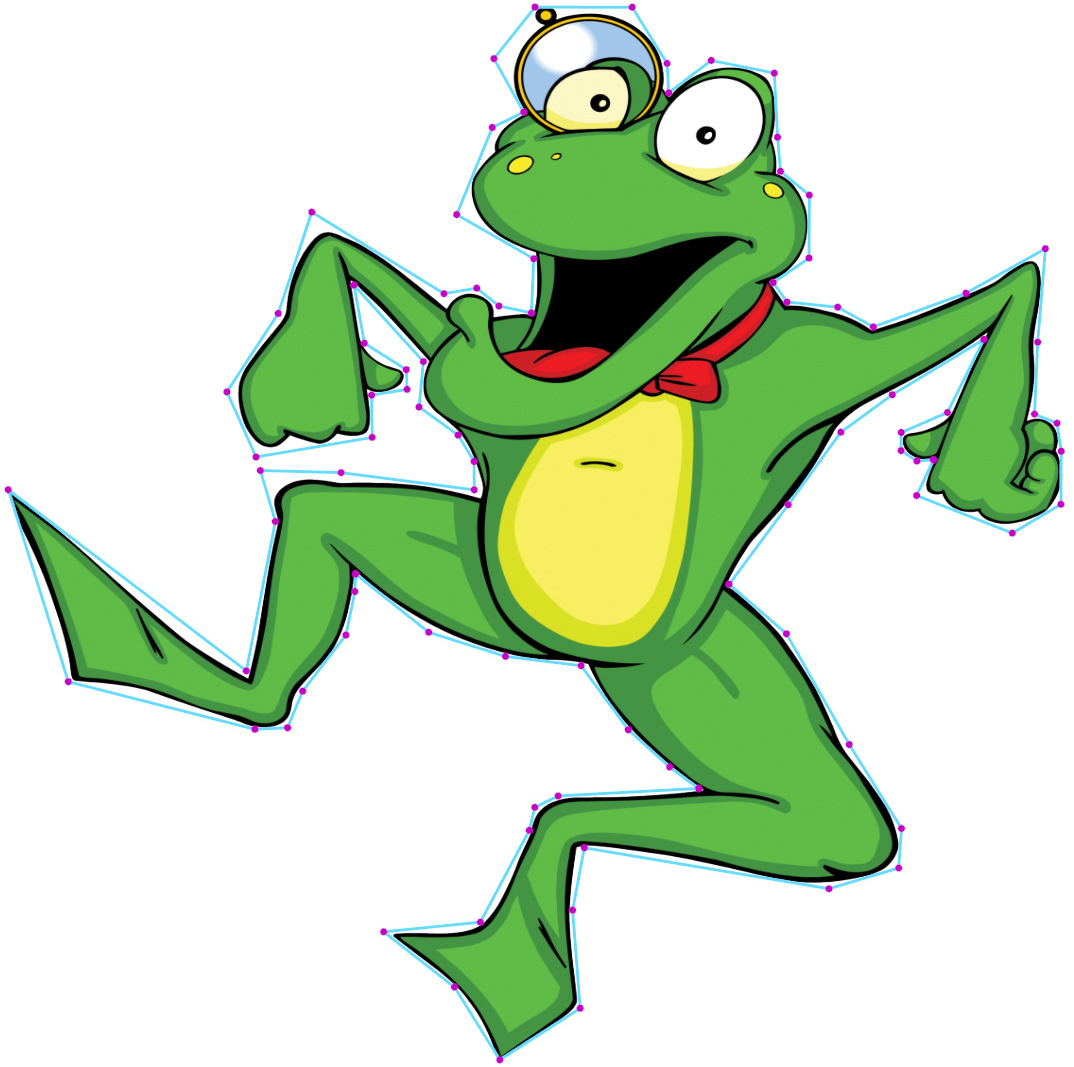}
	  \hfill
	  \includegraphics[width=.3\linewidth]{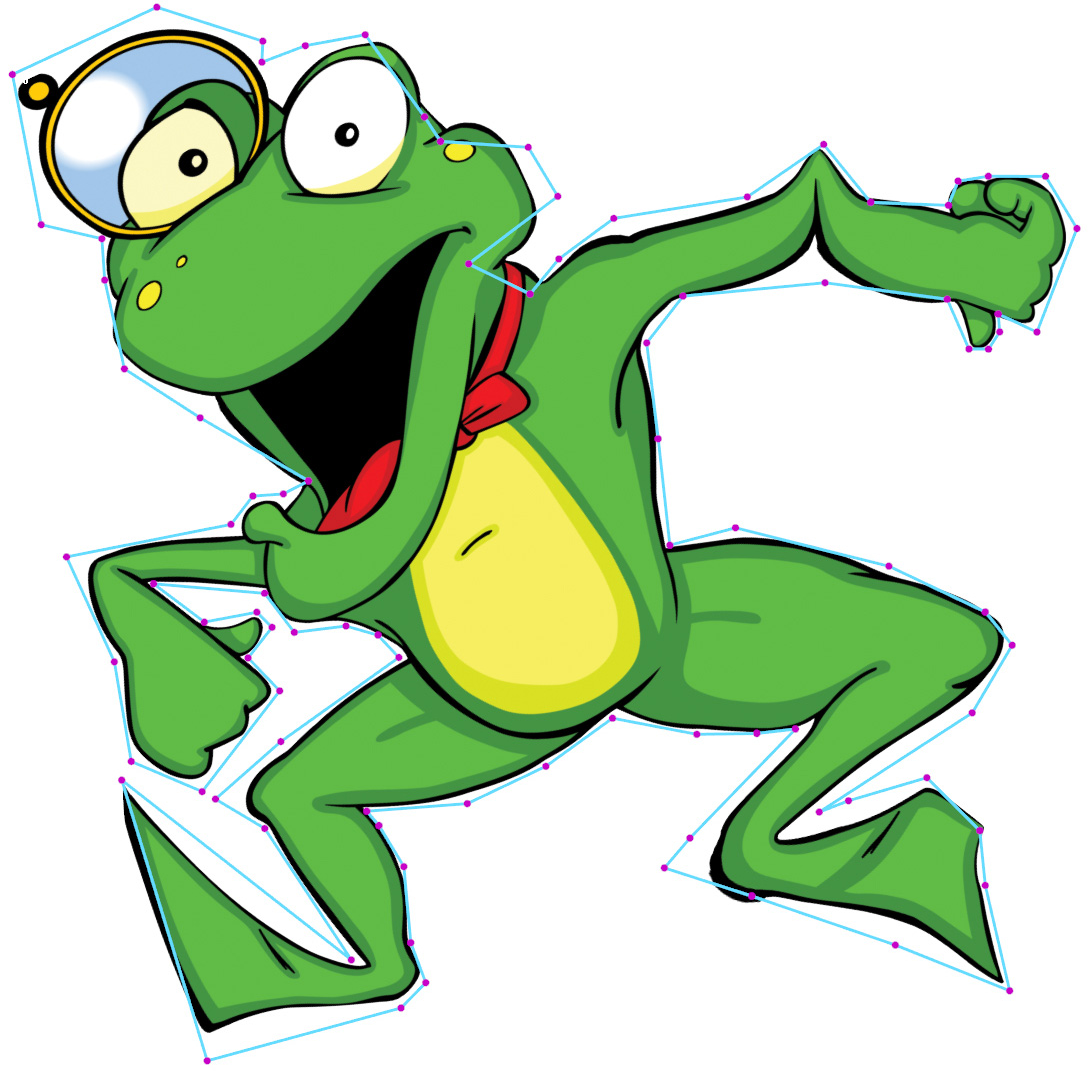}
	  \hfill
	  \includegraphics[width=.3\linewidth]{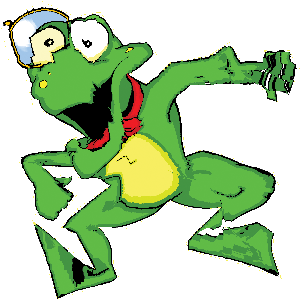}
	\caption{Left: The undeformed shape. Center: A deformation computed using Complex Barycentric Coordinates~\protect\cite{weber2009complex}.
	Right: A deformation computed using Integral Curve Coordinates; the function is the solution to the Laplace equation with the boundary values constrained to zero, subject to the cousin tree constraints. 
	Complex Barycentric Coordinates produce very smooth deformations, but are limited to 2D.}
	\label{fig:frog}
\end{figure}

\begin{figure}[h]
	\centering
	  \includegraphics[height=1in]{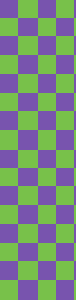}
	  \includegraphics[height=1in]{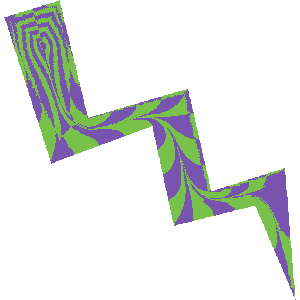}
	  \includegraphics[height=1in]{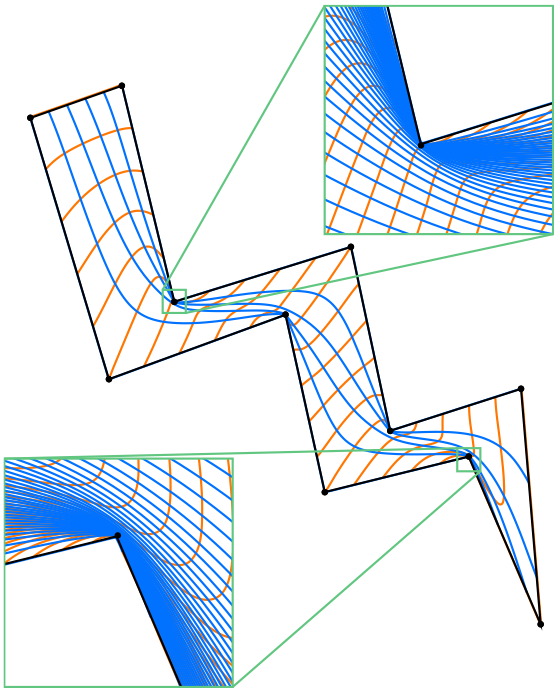}
	\caption{From left to right: The undeformed shape, a deformation computed using Integral Curve Coordinates with a Laplace equation not subject to cousin tree constraints, and the deformation computed using Composite Mean Value Mapping~\protect\cite{schneider2013bijective}.}
	\label{fig:compare_composite}
\end{figure}

\paragraph{Performance}
Our experiments were written in unoptimized Python and executed on a 2 Ghz Intel Core i7.
In all examples, we compute function values on a 50-by-50 discrete grid.
Performance is dominated by per-pixel integral curves tracing.
Our examples contain, on average, 38,850 pixels, and took approximately 30 minutes each.
There are large performance improvements to be obtained by
implementing integral curve tracing in a compiled language
and by parallelization.
Further performance improvements could be obtained by deforming
all points along an integral curve at once, rather than wastefully
re-tracing integral curves for each point along it.
Finally, an approach based on advecting the boundary could efficiently trace
all integral curves at once.

\section{Limitations}
\label{sec:limitations}
\label{sec:expansion}

In our piecewise linear implementation,
compression edges (Section~\ref{sec:integral_lines}) occur when the gradients
on either side of an edge point towards the edge.
Similarly, we call an edge an \emph{expansion edge}
when the gradients on either side point away from it.
%
When tracing integral lines downhill,
multiple integral lines converge at an expansion edge.
This leads to multiple integral lines intersecting the same boundary point;
the points along these integral lines will ``collapse'' as a result of the deformation.
Expansion edges rarely occur for the $L_1$ function values or the Laplace equation,
except near small concavities on the boundary.
Expansion edges are visualized in Figure~\ref{fig:functions} as purple edges,
and typically only occur near concavities of the boundary.
With $C^1$ or $G^1$ function interpolation, expansion edges (and compression edges)
would no longer occur. A looser requirement than continuity is simply that
the gradients on either side of an edge never point away from each other;
a tangible solution for this relaxed condition is unclear.

Our grid discretization of the boundary or cage may lead to problematic boundary gradients
near sharp angles ($< 45^{\circ}$). One solution
is to warp space with a simple ``plaid deformation''
such that (a) grid vertices are positioned exactly at boundary vertices and
(b) sharp angles are non-uniformly scaled and eliminated.

\section{Conclusions and Future Work}

Integral Curve Coordinates provide a new approach for bijective shape deformation
based on tracing the integral curves of functions with one critical point, a maximum.
While the deformations produced by our approach are not as fair as, for example,
Controllable Conformal Maps~\cite{weber2010controllable},
they are bijective in all dimensions.
The fairness of our deformations can be improved in a bijectivity-preserving manner
via the recent, complementary work of Sch{\"u}ller et al.~\cite{schuller2013locally}
and Aigerman and Lipman~\cite{aigerman2013injective}.
%
We believe that fairer functions may be found in our function space.
One approach may be to compute a compatible skeleton for the deformed and undeformed
shapes and treat the entire skeleton as the maximum.

Our approach is restricted to cage-based or boundary deformations.
In the future, we would like to extend our approach to other control structures,
such as points and bone skeletons, which are intuitive to manipulate
and can have far fewer vertices than a cage.
One could trace integral curves of a smoothed distance functions
from the control geometry~\cite{Peng:2004:IMO}.

Our piecewise linear implementation,
while preventing inverted elements ($det(M)<0$ in Equation \ref{eq:jacobian_matrix_discrete}),
does not prevent collapsed elements ($det(M)=0$), which are caused by compression edges.
We would like to address this in the future
with $C^1$ or $G^1$ monotonic interpolation of functions values,
or by simulating the infinitessimal separation of integral curves \cite{EHZ01}
and then perturbing the resulting deformation to correct collapsed elements.

We would also like to explore modifications to the cousin tree constraints. The constraints we compute, while correct, are not unique. Thus, we envisage an iterative procedure that updates the constraints and the function values.
Jacobson \cite{jacobson2013algorithms} explored such an iterative constraint modification scheme in an analogous setting to good effect.
A similar iterative scheme may also remove expansion edges.

Finally, we would like to implement our approach in higher dimensions,
applying it to problems such as the animation of volumetric models.


\section*{Acknowledgements}
We are grateful to Harry Gingold, Jyh-Ming Lien, and Alec Jacobson for fruitful discussions.
Ofir Weber generously provided the frog used in Figure~\ref{fig:frog}.

\bibliographystyle{IEEEtran_noURL}
\bibliography{bib/bibl,bib/coordinates,bib/toposmooth}

\begin{thebibliography}{10}
\def\url#1{}
\csname url@samestyle\endcsname
\providecommand{\newblock}{\relax}
\providecommand{\bibinfo}[2]{#2}
\providecommand{\BIBentrySTDinterwordspacing}{\spaceskip=0pt\relax}
\providecommand{\BIBentryALTinterwordstretchfactor}{4}
\providecommand{\BIBentryALTinterwordspacing}{\spaceskip=\fontdimen2\font plus
\BIBentryALTinterwordstretchfactor\fontdimen3\font minus
  \fontdimen4\font\relax}
\providecommand{\BIBforeignlanguage}[2]{{%
\expandafter\ifx\csname l@#1\endcsname\relax
\typeout{** WARNING: IEEEtran.bst: No hyphenation pattern has been}%
\typeout{** loaded for the language `#1'. Using the pattern for}%
\typeout{** the default language instead.}%
\else
\language=\csname l@#1\endcsname
\fi
#2}}
\providecommand{\BIBdecl}{\relax}
\BIBdecl

\bibitem{allen2003space}
B.~Allen, B.~Curless, and Z.~Popovi\'{c}, ``The space of human body shapes:
  reconstruction and parameterization from range scans,'' \emph{ACM
  Transactions on Graphics}, pp. 587--594, 2003.

\bibitem{xia2010parameterization}
J.~Xia, Y.~He, S.~Han, C.-W. Fu, F.~Luo, and X.~Gu, ``Parameterization of
  star-shaped volumes using green's functions,'' in \emph{Advances in geometric
  modeling and processing}.\hskip 1em plus 0.5em minus 0.4em\relax Springer,
  2010, pp. 219--235.

\bibitem{schuller2013locally}
C.~Sch{\"u}ller, L.~Kavan, D.~Panozzo, and O.~Sorkine-Hornung, ``Locally
  injective mappings,'' in \emph{Computer Graphics Forum}, vol.~32,
  no.~5.\hskip 1em plus 0.5em minus 0.4em\relax Wiley Online Library, 2013, pp.
  125--135.

\bibitem{aigerman2013injective}
N.~Aigerman and Y.~Lipman, ``Injective and bounded distortion mappings in 3d,''
  \emph{ACM Transactions on Graphics (TOG)}, vol.~32, no.~4, p. 106, 2013.

\bibitem{Sederberg:1986:FDS}
\BIBentryALTinterwordspacing
T.~W. Sederberg and S.~R. Parry, ``Free-form deformation of solid geometric
  models,'' in \emph{Proceedings of the 13th Annual Conference on Computer
  Graphics and Interactive Techniques}, ser. SIGGRAPH '86.\hskip 1em plus 0.5em
  minus 0.4em\relax New York, NY, USA: ACM, 1986, pp. 151--160.
  \url{http://doi.acm.org/10.1145/15922.15903}
\BIBentrySTDinterwordspacing

\bibitem{MacCracken:1996:FDL}
\BIBentryALTinterwordspacing
R.~MacCracken and K.~I. Joy, ``Free-form deformations with lattices of
  arbitrary topology,'' in \emph{Proceedings of the 23rd Annual Conference on
  Computer Graphics and Interactive Techniques}, ser. SIGGRAPH '96.\hskip 1em
  plus 0.5em minus 0.4em\relax New York, NY, USA: ACM, 1996, pp. 181--188.
  \url{http://doi.acm.org/10.1145/237170.237247}
\BIBentrySTDinterwordspacing

\bibitem{lipman2012simple}
Y.~Lipman, V.~G. Kim, and T.~A. Funkhouser, ``Simple formulas for
  quasiconformal plane deformations,'' \emph{ACM Transactions on Graphics
  (TOG)}, vol.~31, no.~5, p. 124, 2012.

\bibitem{Magnenat-Thalmann:1988:JLD}
\BIBentryALTinterwordspacing
N.~Magnenat-Thalmann, R.~Laperri\`{e}re, and D.~Thalmann, ``Joint-dependent
  local deformations for hand animation and object grasping,'' in
  \emph{Proceedings of Graphics Interface}, 1988, pp. 26--33.
  \url{http://dl.acm.org/citation.cfm?id=102313.102317}
\BIBentrySTDinterwordspacing

\bibitem{Alexa:2000:ASI}
\BIBentryALTinterwordspacing
M.~Alexa, D.~Cohen-Or, and D.~Levin, ``As-rigid-as-possible shape
  interpolation,'' in \emph{Proceedings of the 27th Annual Conference on
  Computer Graphics and Interactive Techniques}, ser. SIGGRAPH '00.\hskip 1em
  plus 0.5em minus 0.4em\relax New York, NY, USA: ACM Press/Addison-Wesley
  Publishing Co., 2000, pp. 157--164.
  \url{http://dx.doi.org/10.1145/344779.344859}
\BIBentrySTDinterwordspacing

\bibitem{igarashi2005rigid}
T.~Igarashi, T.~Moscovich, and J.~F. Hughes, ``As-rigid-as-possible shape
  manipulation,'' in \emph{ACM Transactions on Graphics (TOG)}, vol.~24,
  no.~3.\hskip 1em plus 0.5em minus 0.4em\relax ACM, 2005, pp. 1134--1141.

\bibitem{joshi2007harmonic}
P.~Joshi, M.~Meyer, T.~DeRose, B.~Green, and T.~Sanocki, ``Harmonic coordinates
  for character articulation,'' in \emph{ACM Transactions on Graphics (TOG)},
  vol.~26, no.~3.\hskip 1em plus 0.5em minus 0.4em\relax ACM, 2007, p.~71.

\bibitem{lipman2008green}
Y.~Lipman, D.~Levin, and D.~Cohen-Or, ``Green coordinates,'' in \emph{ACM
  Transactions on Graphics (TOG)}, vol.~27, no.~3.\hskip 1em plus 0.5em minus
  0.4em\relax ACM, 2008, p.~78.

\bibitem{BenChen:2009:VHM}
M.~Ben-Chen, O.~Weber, and C.~Gotsman, ``Variational harmonic maps for space
  deformation,'' \emph{ACM Transactions on Graphics (TOG)}, vol.~28, no.~3, pp.
  34:1--34:11, 2009.

\bibitem{weber2010controllable}
O.~Weber and C.~Gotsman, ``Controllable conformal maps for shape deformation
  and interpolation,'' \emph{ACM Transactions on Graphics (TOG)}, vol.~29,
  no.~4, p.~78, 2010.

\bibitem{nieto2013cage}
J.~R. Nieto and A.~Sus{\'\i}n, ``Cage based deformations: a survey,'' in
  \emph{Deformation Models}.\hskip 1em plus 0.5em minus 0.4em\relax Springer,
  2013, pp. 75--99.

\bibitem{Sumner:2007:EDS}
R.~W. Sumner, J.~Schmid, and M.~Pauly, ``Embedded deformation for shape
  manipulation,'' in \emph{ACM Transactions on Graphics (TOG)}, vol.~26,
  no.~3.\hskip 1em plus 0.5em minus 0.4em\relax ACM, 2007.

\bibitem{Jacobson:2011:BBW}
\BIBentryALTinterwordspacing
A.~Jacobson, I.~Baran, J.~Popovi\'{c}, and O.~Sorkine, ``Bounded biharmonic
  weights for real-time deformation,'' \emph{ACM Trans. Graph.}, vol.~30,
  no.~4, pp. 78:1--78:8, Jul. 2011.
  \url{http://doi.acm.org/10.1145/2010324.1964973}
\BIBentrySTDinterwordspacing

\bibitem{wachspress1975rational}
E.~L. Wachspress, \emph{A rational finite element basis}.\hskip 1em plus 0.5em
  minus 0.4em\relax Academic Press New York, 1975, vol. 114.

\bibitem{floater2003mean}
M.~S. Floater, ``Mean value coordinates,'' \emph{Computer Aided Geometric
  Design}, vol.~20, no.~1, pp. 19--27, 2003.

\bibitem{ju2005mean}
T.~Ju, S.~Schaefer, and J.~Warren, ``Mean value coordinates for closed
  triangular meshes,'' in \emph{ACM Transactions on Graphics (TOG)}, vol.~24,
  no.~3.\hskip 1em plus 0.5em minus 0.4em\relax ACM, 2005, pp. 561--566.

\bibitem{belyaev2006transfinite}
A.~Belyaev, ``On transfinite barycentric coordinates,'' in \emph{Proceedings of
  the fourth Eurographics symposium on Geometry processing}.\hskip 1em plus
  0.5em minus 0.4em\relax Eurographics Association, 2006, pp. 89--99.

\bibitem{warren2007barycentric}
J.~Warren, S.~Schaefer, A.~N. Hirani, and M.~Desbrun, ``Barycentric coordinates
  for convex sets,'' \emph{Advances in computational mathematics}, vol.~27,
  no.~3, pp. 319--338, 2007.

\bibitem{hormann2008maximum}
K.~Hormann and N.~Sukumar, ``Maximum entropy coordinates for arbitrary
  polytopes,'' in \emph{Computer Graphics Forum}, vol.~27, no.~5.\hskip 1em
  plus 0.5em minus 0.4em\relax Wiley Online Library, 2008, pp. 1513--1520.

\bibitem{weber2009complex}
O.~Weber, M.~Ben-Chen, and C.~Gotsman, ``Complex barycentric coordinates with
  applications to planar shape deformation,'' in \emph{Computer Graphics
  Forum}, vol.~28, no.~2.\hskip 1em plus 0.5em minus 0.4em\relax Wiley Online
  Library, 2009, pp. 587--597.

\bibitem{manson2010moving}
J.~Manson and S.~Schaefer, ``Moving least squares coordinates,'' in
  \emph{Computer Graphics Forum}, vol.~29, no.~5.\hskip 1em plus 0.5em minus
  0.4em\relax Wiley Online Library, 2010, pp. 1517--1524.

\bibitem{weber2011complex}
O.~Weber, M.~Ben-Chen, C.~Gotsman, and K.~Hormann, ``A complex view of
  barycentric mappings,'' in \emph{Computer Graphics Forum}, vol.~30,
  no.~5.\hskip 1em plus 0.5em minus 0.4em\relax Wiley Online Library, 2011, pp.
  1533--1542.

\bibitem{jacobson2013algorithms}
A.~Jacobson, ``Algorithms and interfaces for real-time deformation of 2d and 3d
  shapes,'' Ph.D. dissertation, Diss., Eidgen{\"o}ssische Technische Hochschule
  ETH Z{\"u}rich, Nr. 21189, 2013.

\bibitem{Xu:2011:ETG}
\BIBentryALTinterwordspacing
Y.~Xu, R.~Chen, C.~Gotsman, and L.~Liu, ``Embedding a triangular graph within a
  given boundary,'' \emph{Comput. Aided Geom. Des.}, vol.~28, no.~6, pp.
  349--356, Aug. 2011.  \url{http://dx.doi.org/10.1016/j.cagd.2011.07.001}
\BIBentrySTDinterwordspacing

\bibitem{schneider2013bijective}
T.~Schneider, K.~Hormann, and M.~S. Floater, ``Bijective composite mean value
  mappings,'' in \emph{Computer Graphics Forum}, vol.~32, no.~5.\hskip 1em plus
  0.5em minus 0.4em\relax Wiley Online Library, 2013, pp. 137--146.

\bibitem{lipman2012bounded}
Y.~Lipman, ``Bounded distortion mapping spaces for triangular meshes,''
  \emph{ACM Transactions on Graphics (TOG)}, vol.~31, no.~4, p. 108, 2012.

\bibitem{kraevoy2003matchmaker}
V.~Kraevoy, A.~Sheffer, and C.~Gotsman, ``Matchmaker: constructing constrained
  texture maps,'' in \emph{ACM Transactions on Graphics (TOG)}, vol.~22,
  no.~3.\hskip 1em plus 0.5em minus 0.4em\relax ACM, 2003, pp. 326--333.

\bibitem{seo2010constrained}
H.~Seo and F.~Cordier, ``Constrained texture mapping using image warping,'' in
  \emph{Computer Graphics Forum}, vol.~29, no.~1.\hskip 1em plus 0.5em minus
  0.4em\relax Wiley Online Library, 2010, pp. 160--174.

\bibitem{lee2008texture}
T.-Y. Lee, S.-W. Yen, and I.-C. Yeh, ``Texture mapping with hard constraints
  using warping scheme,'' \emph{Visualization and Computer Graphics, IEEE
  Transactions on}, vol.~14, no.~2, pp. 382--395, 2008.

\bibitem{Weinkauf:2010:TBS}
T.~Weinkauf, Y.~Gingold, and O.~Sorkine, ``Topology-based smoothing of 2{D}
  scalar fields with {C1}-continuity,'' \emph{Computer Graphics Forum
  (proceedings of EuroVis)}, vol.~29, no.~3, pp. 1221--1230, 2010.

\bibitem{Jacobson:2012:SSA}
A.~Jacobson, T.~Weinkauf, and O.~Sorkine, ``Smooth shape-aware functions with
  controlled extrema,'' \emph{Computer Graphics Forum (Proc. SGP)}, vol.~31,
  no.~5, pp. 1577--1586, July 2012.

\bibitem{Zom01}
A.~Zomorodian, ``Computing and comprehending topology: Persistence and
  hierarchical {M}orse complexes,'' Ph.D. dissertation, University of Illinois
  at Urbana-Champaign, 2001.

\bibitem{beatson1985monotonicity}
R.~Beatson and Z.~Ziegler, ``Monotonicity preserving surface interpolation,''
  \emph{SIAM journal on numerical analysis}, vol.~22, no.~2, pp. 401--411,
  1985.

\bibitem{floater1998tensor}
M.~S. Floater and J.~M. Pe{\~n}a, ``Tensor-product monotonicity preservation,''
  \emph{Advances in Computational Mathematics}, vol.~9, no. 3-4, pp. 353--362,
  1998.

\bibitem{Schmidt:1992:PMS}
\BIBentryALTinterwordspacing
J.~W. Schmidt, ``Positive, monotone, and {S}-convex $c^1$-interpolation on
  rectangular grids,'' \emph{Computing}, vol.~48, no. 3--4, pp. 363--371, Mar.
  1992.  \url{http://dx.doi.org/10.1007/BF02238643}
\BIBentrySTDinterwordspacing

\bibitem{Carlson:1989:AMP}
\BIBentryALTinterwordspacing
R.~E. Carlson and F.~N. Fritsch, ``\BIBforeignlanguage{English}{An algorithm
  for monotone piecewise bicubic interpolation},''
  \emph{\BIBforeignlanguage{English}{SIAM Journal on Numerical Analysis}},
  vol.~26, no.~1, pp. pp. 230--238, 1989.
  \url{http://www.jstor.org/stable/2157715}
\BIBentrySTDinterwordspacing

\bibitem{Blum:1967:TEN}
H.~Blum, ``A transformation for extracting new descriptors of shape,'' in
  \emph{Models for the Perception of Speech and Visual Form. Proceedings of a
  Symposium}, W.~Wathen-Dunn, Ed.\hskip 1em plus 0.5em minus 0.4em\relax
  Cambridge MA: MIT Press, Nov 1967, pp. 362--380.

\bibitem{edelsbrunner1990simulation}
H.~Edelsbrunner and E.~P. M{\"u}cke, ``Simulation of simplicity: a technique to
  cope with degenerate cases in geometric algorithms,'' \emph{ACM Transactions
  on Graphics (TOG)}, vol.~9, no.~1, pp. 66--104, 1990.

\bibitem{Yap:GCT:1990}
C.-K. Yap, ``A geometric consistency theorem for a symbolic perturbation
  scheme,'' \emph{J. Comput. Systems Sci.}, vol.~40, pp. 2--18, 1990.

\bibitem{Yap:STG:1990}
------, ``Symbolic treatment of geometric degeneracies,'' \emph{J. Symbolic
  Comput.}, vol.~10, pp. 349--370, 1990.

\bibitem{Jacobson:2013:BMG}
\BIBentryALTinterwordspacing
A.~Jacobson, ``Bijective mappings with generalized barycentric coordinates: A
  counterexample,'' \emph{Journal of Graphics Tools}, vol.~17, no. 1--2, pp.
  1--4, 2013.  \url{http://dx.doi.org/10.1080/2165347X.2013.842511}
\BIBentrySTDinterwordspacing

\bibitem{Joshi:2007:HCC}
P.~Joshi, M.~Meyer, T.~DeRose, B.~Green, and T.~Sanocki, ``Harmonic coordinates
  for character articulation,'' \emph{ACM Trans. Graph.}, vol.~26, no.~3, Jul.
  2007.

\bibitem{Peng:2004:IMO}
J.~Peng, D.~Kristjansson, and D.~Zorin, ``Interactive modeling of topologically
  complex geometric detail,'' \emph{ACM Transactions on Graphics}, vol.~23,
  no.~3, pp. 635--643, 2004.

\bibitem{EHZ01}
H.~Edelsbrunner, J.~Harer, and A.~Zomorodian, ``Hierarchical {M}orse complexes
  for piecewise linear 2-manifolds,'' in \emph{Proc. 17th Ann. ACM Sympos.
  Comput. Geom.}, 2001, pp. 70--79.

\end{thebibliography}

\appendix
\noindent
In the following, the domain is a regular $n$-dimensional grid with edges parallel to the $x_1,x_2,x_3,\hdots,x_n$ axes.
We call a subdomain \emph{ball-like} if the volume enclosed by its cells is contractible. 

\begin{theorem}
\label{thm:cousin_correct}
Let $D$ be a subdomain that is ball-like. Let $T$ be an associated cousin tree. Assume all nodes on the boundary of the subdomain have value greater than all grid neighbors outside the ball (including along imaginary diagonal grid edges). 
Then by rooting $T$ at the maximum and assigning monotonically decreasing (unique) values to internal nodes of $T$ along the path from the root (preserving the values at the root and the leaf), there can be no grid maximum, minimum, or saddles at any tree node
with 4- or 8-connectivity. 
(In dimensions $> 2$, we assume monotonic interpolation within a hypercube;
i.e.\ $2n$-connectivity.)
\end{theorem}

\begin{lemma}
\label{lemma:child_continues}
A node $v$ of cousin tree $T$ on subdomain $D$ with parent in the $+x_i$ grid direction has a child in the $-x_i$ grid direction, unless the neighbor in the $-x_i$ grid direction is outside $D$.
\end{lemma}

\begin{proof}
Without loss of generality, assume node $v$ has its tree parent $u$ in the $+x_2$ grid direction.  Let $w$ be the node in the $-x_1$ grid direction.  Assume $v$ is not $w$'s tree parent.

\begin{center}
  \includegraphics[width=.25\linewidth]{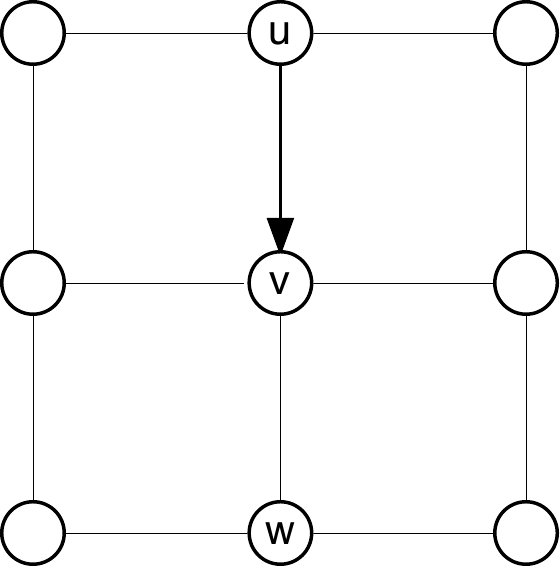}
\end{center}

The cousin rule tells us that $v$'s tree parent $u$ and $w$'s tree parent must be grid neighbors.  Yet the only grid neighbor of $w$ within 1 grid edge of $v$'s parent of $u$ is $v$.  Thus we have reached a contradiction, and $v$ must be $w$'s parent.
\end{proof}

\begin{lemma}
\label{lemma:connected_components}
For any non-boundary tree node $v$, its grid neighbors (even with the addition of diagonal edges on the hypercube face planes) with greater value all belong to the same connected component (partitioned according to greater/less than $v$'s
value) and its grid neighbors (even with the addition of diagonal edges)
with lesser value all belong to a second connected component.
\end{lemma}

\begin{proof}
Without loss of generality, assume $v$ has parent $u$ in the $+x_2$ direction.  We now consider two cases, $v$ is on the boundary of $D$ and $v$ is not on the boundary of $D$.

	\paragraph{Case {$v$ is not on the boundary of $D$}:}
	
	It follows from Lemma \ref{lemma:child_continues} that $v$ has a child $w$ in the $-x_2$ direction.  As paths are monotonic, $Value(u) > Value(v) > Value(w)$.  Since we wish to prove that $v$ will have exactly two connected components partitioned according to greater/less than $v$'s value, we can restrict our examination to grid neighbors of $u,v,w$ in the $+x_1$ direction, again without loss of generality; the greater value connected component will have to include $u$ and the lesser value connected component will have to include $w$.  We now consider two subcases, {$v$ is not the tree parent of its $x_1$-direction grid neighbor $b$} and {$v$ is the tree parent of $b$}. To clarify, when determining connected components, neighbors along diagonal edges may be considered; however, tree edges are always grid edges, and neighbors in any sense except when computing connected components are only along grid edges.

	\begin{center}  \includegraphics[width=.8\linewidth]{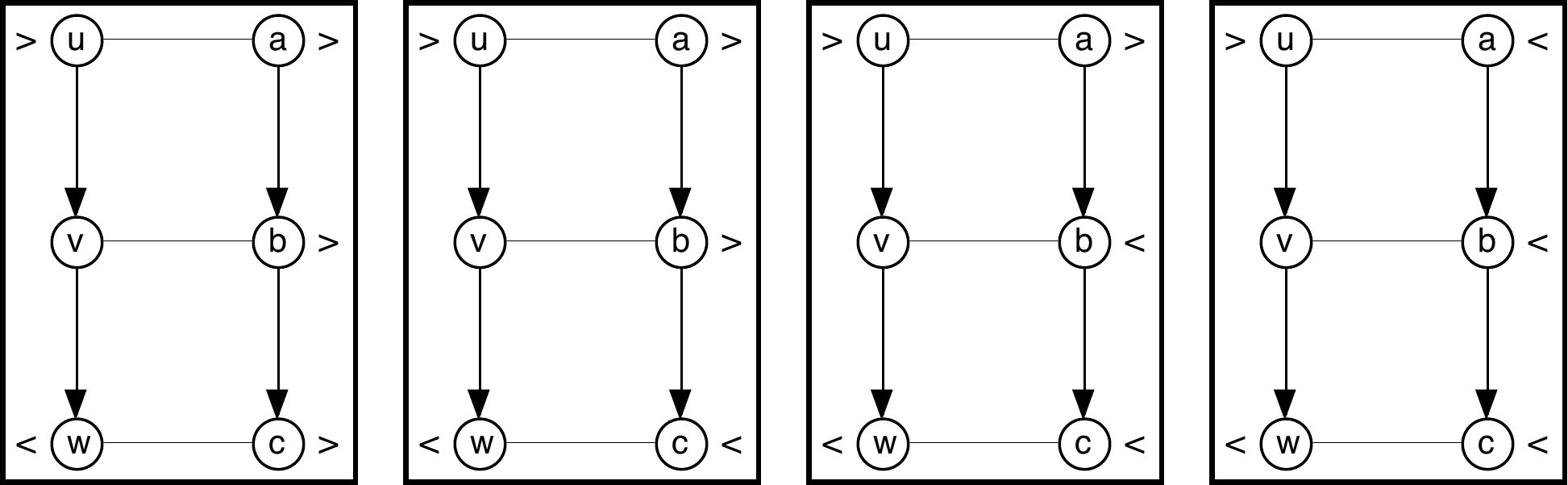}
	\end{center}

	Because of the cousin rule, $b$'s parent must be $a$, since $a$ is the only grid neighbor of $b$ within 1 grid edge of $v$'s parent of $v$.  Lemma \ref{lemma:child_continues} tells us that $c$ must be the tree child of $b$.  Due to the monotonicity of values along tree paths, $Value(u) > Value(v) > Value(w)$ and $Value(a) > Value(b) > Value(c)$.  The following diagram depicts all possible greater/less than relationships between $a,b,c$ and $v$.

	\begin{center}
	\includegraphics[width=.15\linewidth]{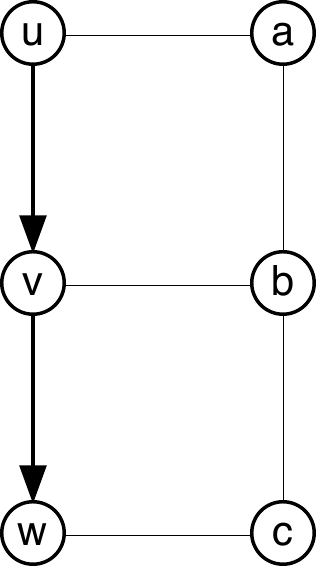}
	\end{center}

	As we can see, in all possibilities, there are exactly two connected components partitioned according to greater/less than $v$'s value.  The component with values greater obviously contains $v$'s tree parent $u$, and the component with values lesser obviously contains $v$'s tree child $w$.

		\emph{Subcase:} $v$ is the tree parent of $b$.

		\begin{center}
		\includegraphics[width=.15\linewidth]{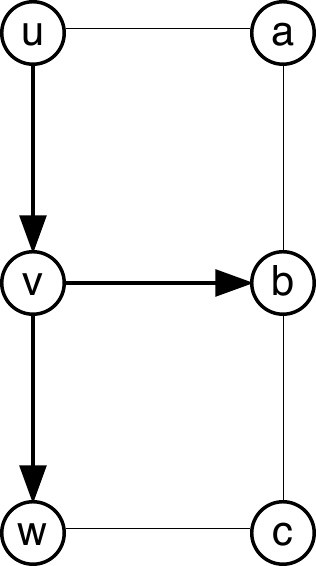}
		\end{center}

		The cousin tree imposes the following restrictions on $a$ and $c$.  $a$ cannot be the tree child of $b$ as $u$ and a would not be tree cousins.  Since $b$ cannot be $a$'s tree parent or its tree child, $b$ and $a$ must be tree cousins.  Therefore $a$'s tree parent must be a grid neighbor of $v$.  The only possibility is $u$.  Since $c$ cannot be the tree parent of $b$ (resp. $w$), it must be the tree child or cousin of $b$ (resp. $w$).  Therefore $c$ must be the tree child of one and the tree cousin of the other.  In either case, the monotonicity of values along tree paths implies $Value(v) > Value(c)$ as well as $Value(v) > Value(b)$, $Value(v) > Value(w)$, and $Value(u) > Value(v)$.  The value of $a$ may be greater than or less than $v$, but in both cases there are exactly two connected components (partitioned according to greater/less than $v$'s value).  As in the earlier subcase, the component with values greater obviously contains $v$'s tree parent $u$, and the component with values lesser obviously contains $v$'s tree child $w$.

	\paragraph{Case {$v$ is on the boundary of $D$}:}
	
	If $v$'s grid neighbor $w$ in the $-x_2$ direction is in $D$, then Lemma \ref{lemma:child_continues} applies and $v$ must be the tree parent of $w$ and hence $Value(v) > Value(w)$.  Otherwise, $w$ is not in $D$, but by assumption $Value(v) > Value(w)$.  Thus $v$'s neighbor in the $+x_2$ direction has value greater than $v$, and $v$'s neighbor in the $-x_2$ direction has value less than $v$.  Following the same argument as in Case {$v$ is not on the boundary of $D$}, we can again restrict our examination to grid neighbors of $u,v,w$ in the $+x_2$ direction (without loss of generality).

	\begin{center}
	\includegraphics[width=.15\linewidth]{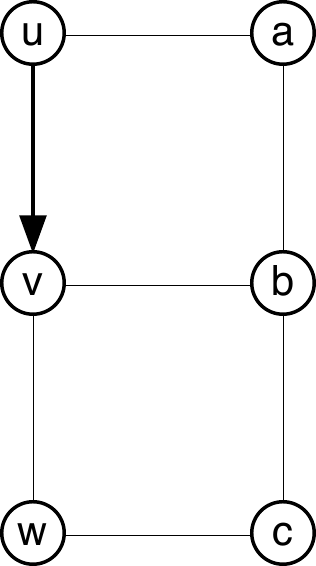}
	\end{center}

	There are 33 valid cousin tree configurations among the $2^{4}$ possible boundary conditions given $a,b,c,w$ can each be outside $D$.  They are presented in Figure \ref{fig:cousin_boundary_conditions}.
		
	In every case there are exactly two connected components partitioned according to greater/less than $v$'s value.  The component with values greater obviously contains $v$'s tree parent $u$, and the component with values lesser obviously $w$.
\end{proof}

\begin{proof}
Lemma \ref{lemma:connected_components} tells us that any internal tree node $v$ has one connected component of nodes with value greater than $v$'s and another connected component with values lesser (0 ``folds'').  It follows directly then that $v$ cannot be a minimum, maximum, or saddle.
\end{proof}

\begin{theorem}
\label{thm:bfs_correct}
A breadth-first search (BFS) in a ball-like domain $D$ where at each BFS generation all $x_1$, then all $x_2$, then all $x_3$, and so on, grid edges are explored in order constructs a cousin tree.
\end{theorem}

\begin{proof}
By induction.  Let $n$ represent the number of BFS generations of growth from the root node.  After step $n$, nodes reached at BFS generation $n-1$ have all of their grid neighbors in the BFS tree or outside $D$; grid neighbors in the BFS tree will be shown to satisfy the cousin tree definition.

After step $n=1$, the root node has become the parent of all grid neighbors.  The root node satisfies the cousin tree constraints by being the BFS parent of all grid neighbors.

Assume true for $n=k$.  Nodes reached at BFS generation $k-1$ have all grid neighbors also in the BFS tree with cousin tree definition satisfied or outside $D$ for BFS generation $\leq k-1$ nodes.

We wish to show that after step $n=k+1$ all nodes reached at BFS generation $k$ now also have all their grid neighbors in the BFS tree satisfying the cousin tree definition or outside $D$.

Consider a grid node $v$ reached at BFS generation $k$.  Let $u$ be the BFS parent of $v$.  $u$ was thus reached at BFS generation $k-1$.  Consider the following arbitrary axis-aligned figure of $v$:

\begin{center}
	\includegraphics[width=.3\linewidth]{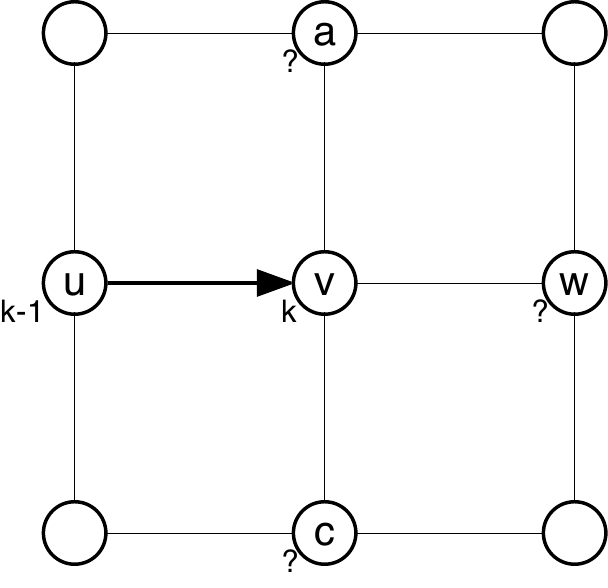}
\end{center}

After step $n=k+1$ nodes $a,c,w$ are also in the BFS tree.  We aim to show that $a,c,w$ are either the BFS children of $v$, tree cousins of $v$ in the BFS tree, or outside $D$.

If $w$ (or $a,c$) is outside $D$, then it is of no concern to us.  If $w$ (or $a,c$) is inside $D$, then it must have been reached at step $k-1$, $k$, or $k+1$.  $w$ (or $a,c$) cannot have been reached before step $k-1$ since it would have been able to reach $v$ at step $k-1$.  $w$ (or $a,c$) must be reached before step $k+2$ since $v$ can reach it at step $k+1$.

Suppose node $w$ was reached at generation $k-1$.

\begin{center}
	\includegraphics[width=.3\linewidth]{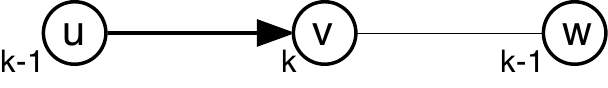}
\end{center}

The inductive hypothesis tells us that $w$ has cousin tree relationships to all its grid neighbors, including $v$.  Yet $w$ is not the BFS parent, BFS child, or BFS cousin of $v$ (since the BFS parent of $w$ cannot be a grid neighbor of $u$).  This contradicts $w$ having been reached at BFS generation $k-1$.

Suppose node $w$ was reached at generation $k$.

\begin{center}
	\includegraphics[width=.3\linewidth]{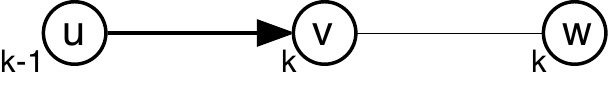}
\end{center}

This too is impossible since $v,w$ are neighbors and cannot have the same taxicab distance to the BFS root in our ball-like domain $D$.

This leaves us with $w$ reached at generation $k+1$.

\begin{center}
	\includegraphics[width=.3\linewidth]{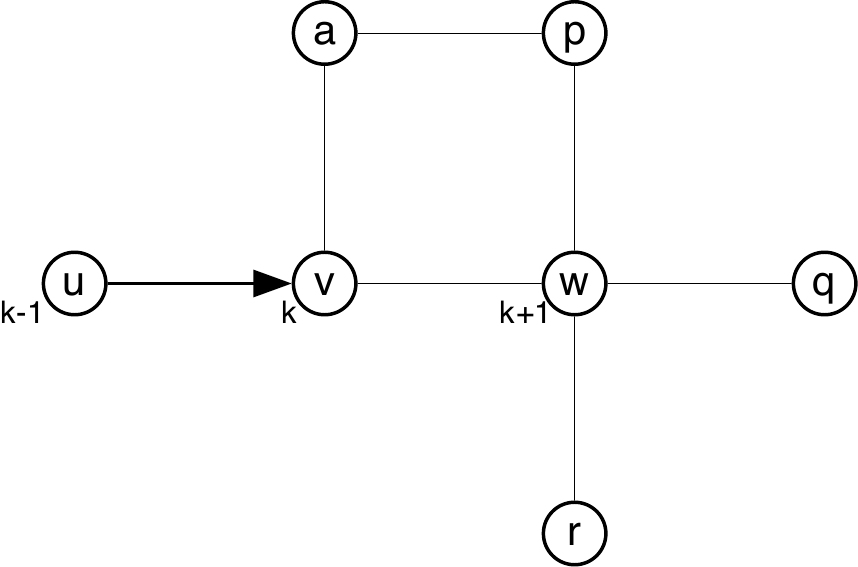}
\end{center}

We first consider $p$ as the BFS parent of $w$.  Then $p$ must have been reached at BFS generation $k$.  Furthermore, if $p,v$ were both reached at BFS generation $k$ and $w$ was reached at BFS generation $k+1$, then $a$ must be in $D$ and have generation $k-1$, since $D$ is ball-like and distance is taxicab.  But $p$ as the BFS parent of $w$ contradicts the BFS $x_1, x_2, x_3, \hdots$ growth ordering implied by $u$ as the BFS parent of $v$, since $u$ is the BFS parent of $v$ instead of $a$.
The same argument prevents $r$ as the BFS parent of $w$.
Consider $q$ as the BFS parent of $w$.  Then $q$ and $v$ were reached at the same BFS generation while $w$ between them was reached at a later BFS generation.  This is impossible in our ball-like domain $D$ since BFS generations correspond to minimal taxicab distance.
The only remaining possibility is $v$ as the BFS parent of $w$.  This is a valid cousin tree relationship.

Suppose node $a$ was reached at generation $k-1$.

\begin{center}
	\includegraphics[width=.3\linewidth]{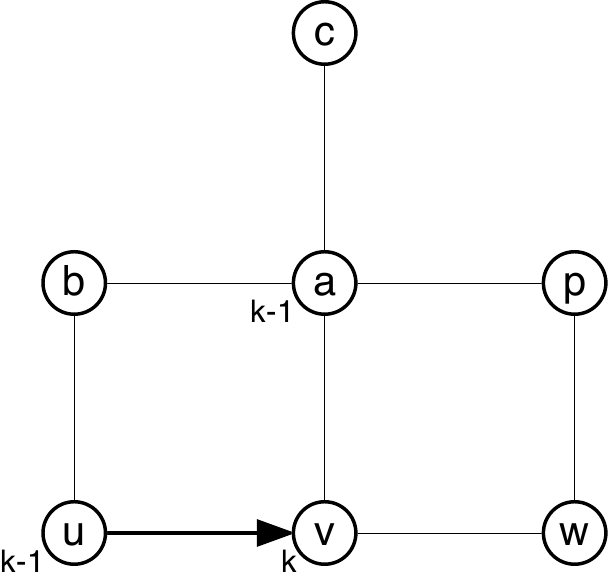}
\end{center}

The inductive hypothesis tells us that $a$ has cousin tree relationships to all its grid neighbors, including $v$.

Suppose node $a$ was reached at generation $k$.

\begin{center}
	\includegraphics[width=.3\linewidth]{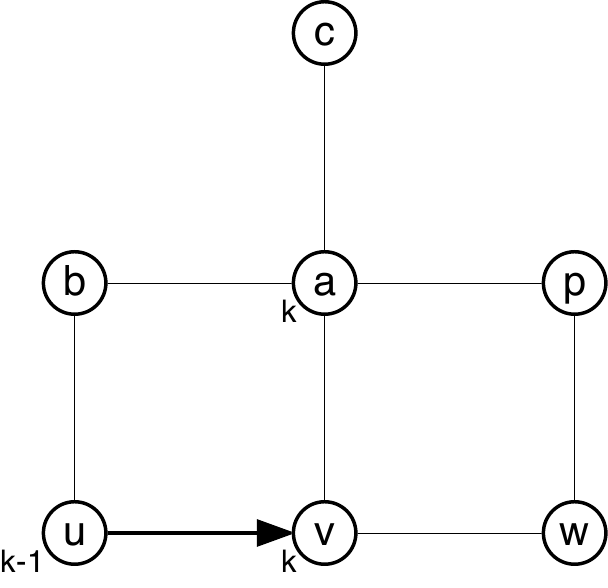}
\end{center}

This too is impossible since $v,a$ are neighbors and cannot have the same taxicab distance to the BFS root in our ball-like domain $D$.

Suppose node $a$ was reached at generation $k+1$.

\begin{center}
	\includegraphics[width=.3\linewidth]{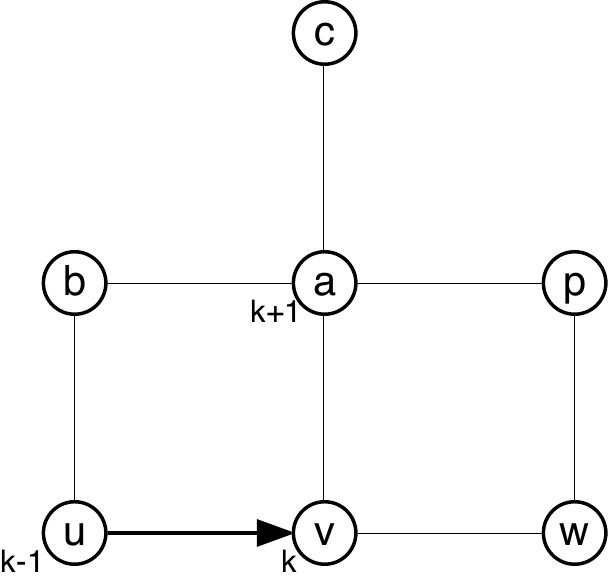}
\end{center}

If node $p$ is the BFS parent of $a$, then $p$ was reached at BFS generation $k$.  We then have the following diagram.

\begin{center}
	\includegraphics[width=.3\linewidth]{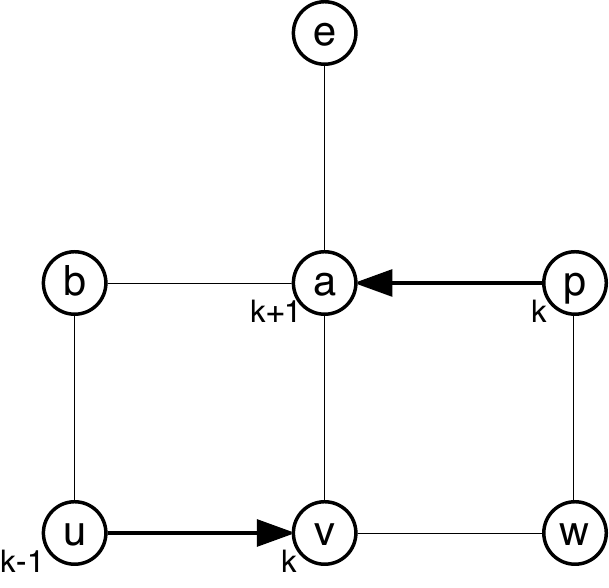}
\end{center}

But nodes $v$ and $p$ cannot have the same taxicab distance (as evidenced by their BFS generations) to the BFS root in ball-like $D$ if nodes $a$ and $u$ differ in taxicab distance to the BFS root by 2.
The same argument prevents $e$ from being the BFS parent of $a$.
The only remaining possibilities are
	$v$ as the BFS parent of $a$
and
	$b$ as the BFS parent of $a$,
both of which have valid cousin tree relationships with $v$.

The analysis of the relationship between node $v$ and node $c$ follows exactly the same argument as between nodes $v$ and $a$.

	\begin{figure}[h]
	\centering

	\includegraphics[width=.8\linewidth]{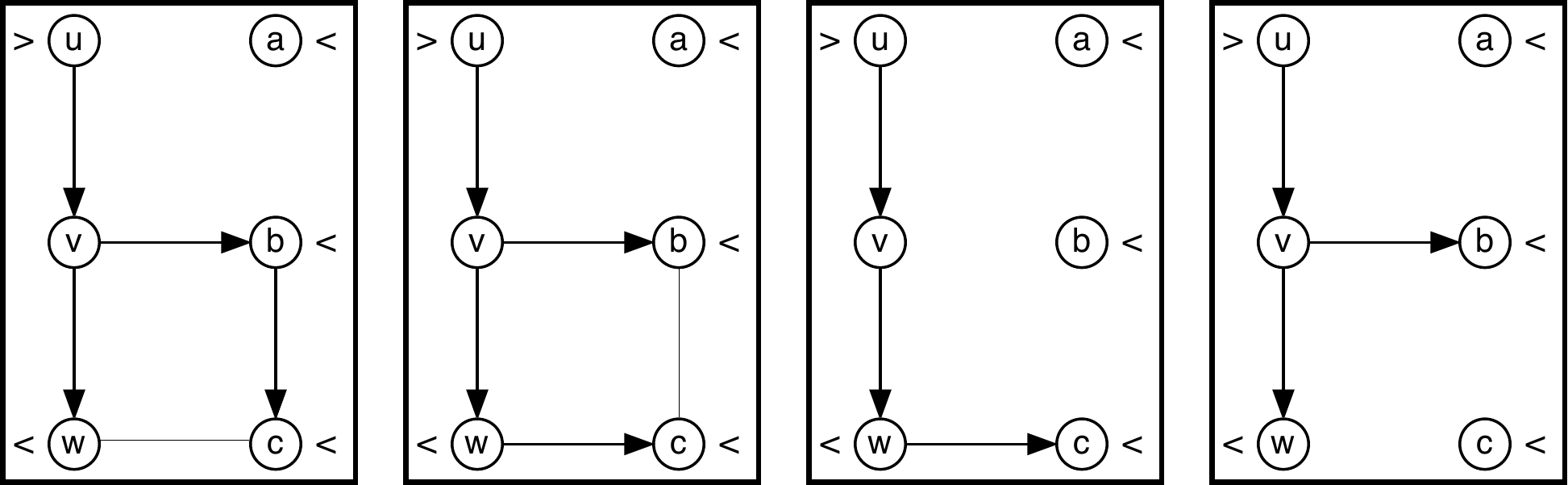}\hfill
	\includegraphics[width=.8\linewidth]{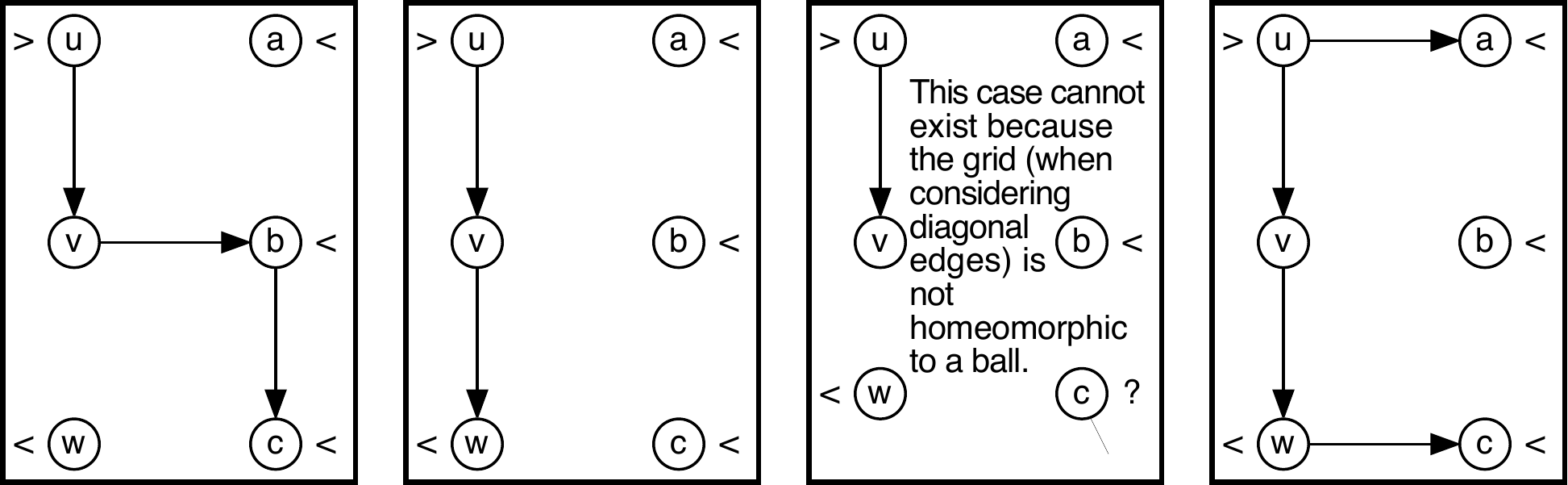}
	\hfill
	\includegraphics[width=.8\linewidth]{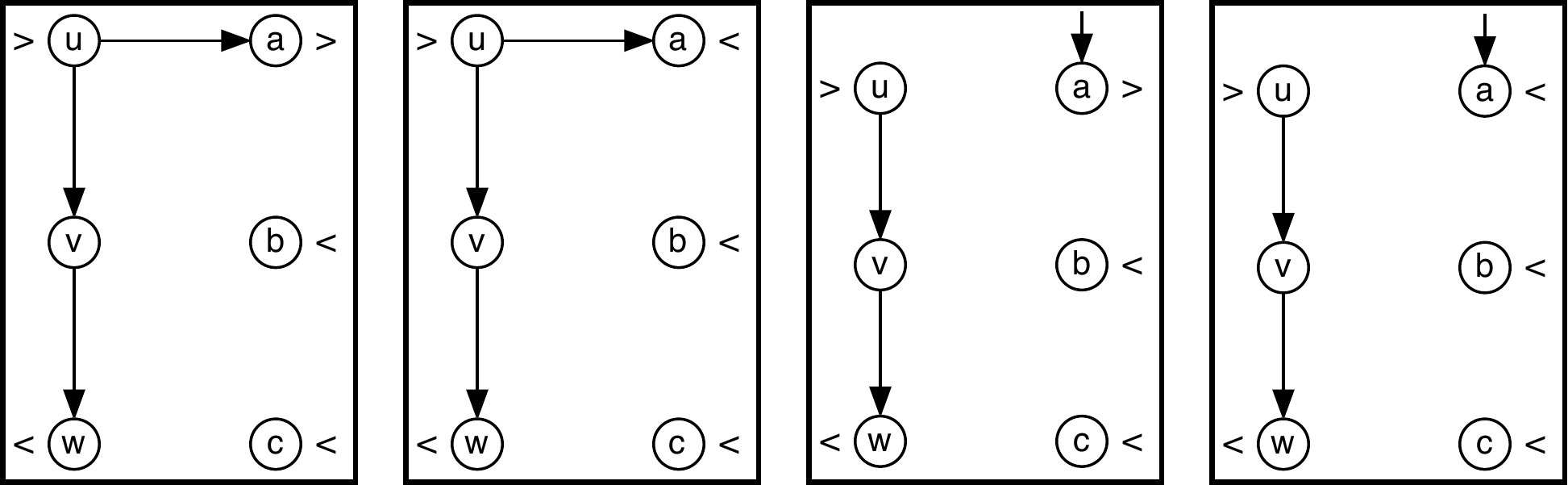}
	\hfill
	\includegraphics[width=.8\linewidth]{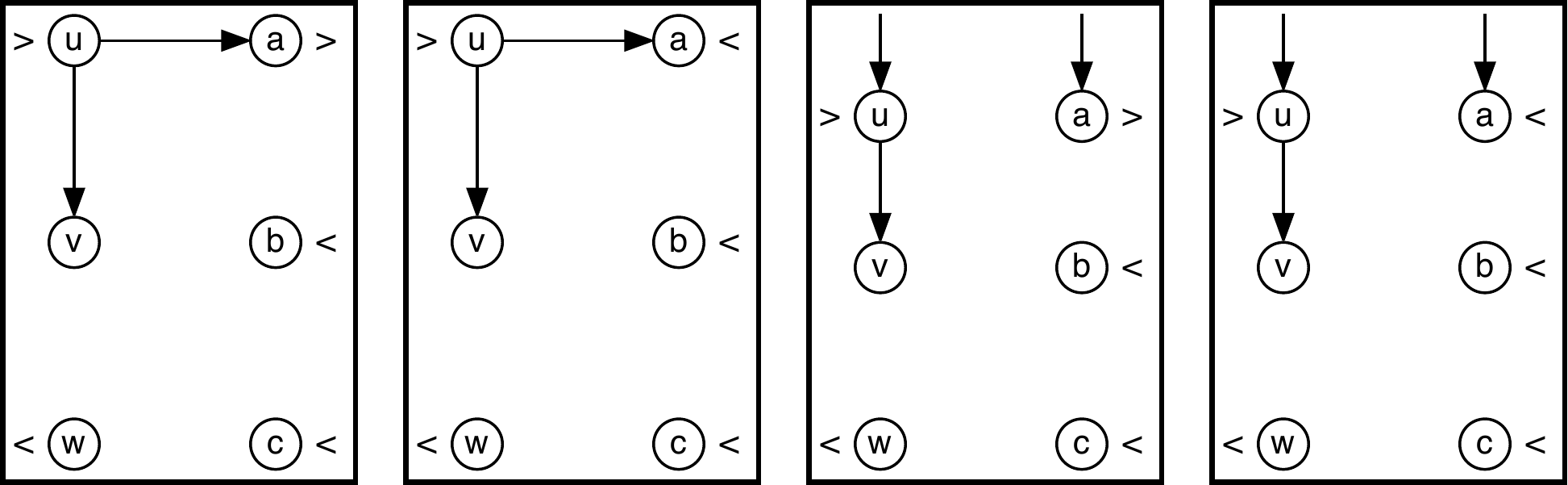}
	\hfill
	\includegraphics[width=.8\linewidth]{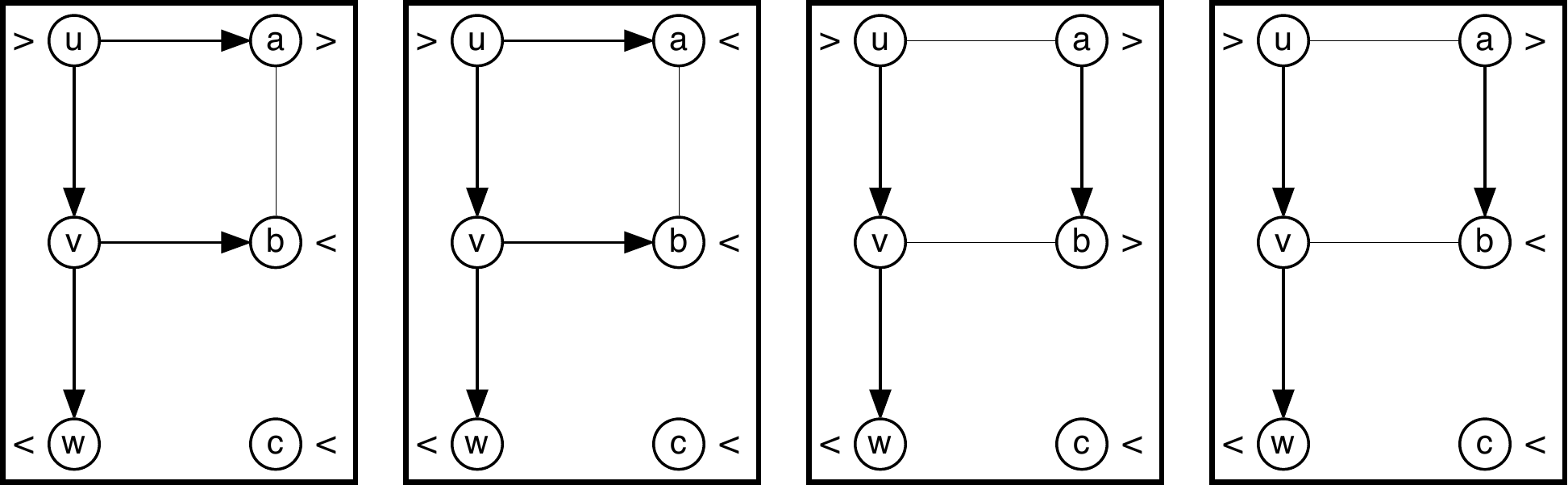}
	\hfill
	\includegraphics[width=.8\linewidth]{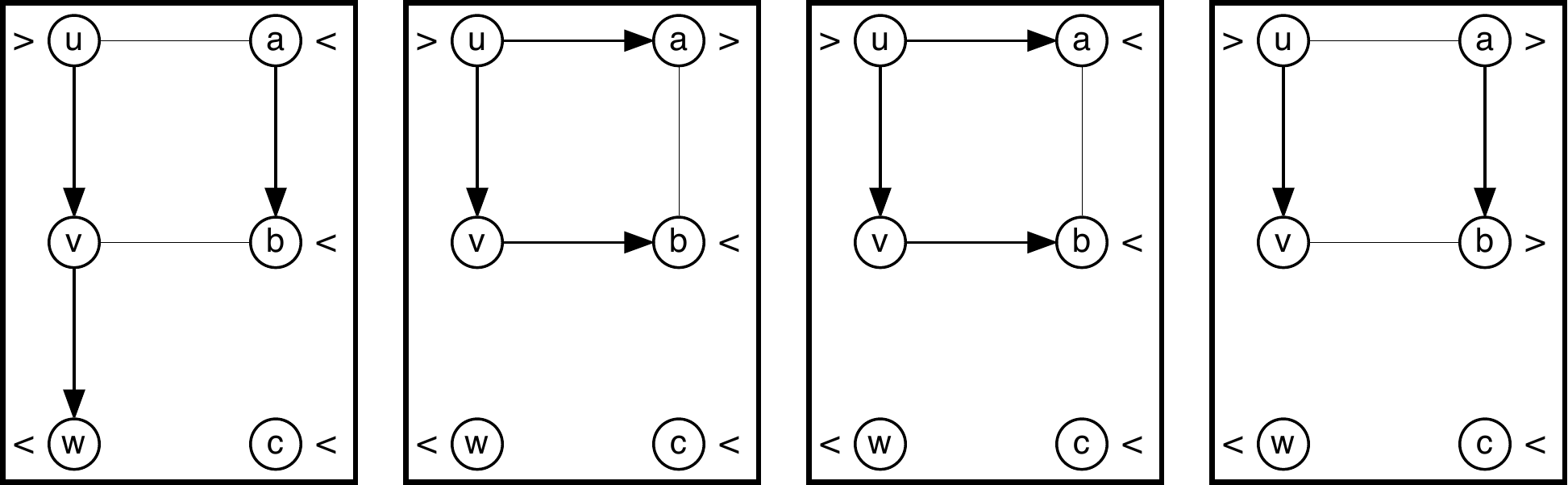}
	\hfill
	\includegraphics[width=.8\linewidth]{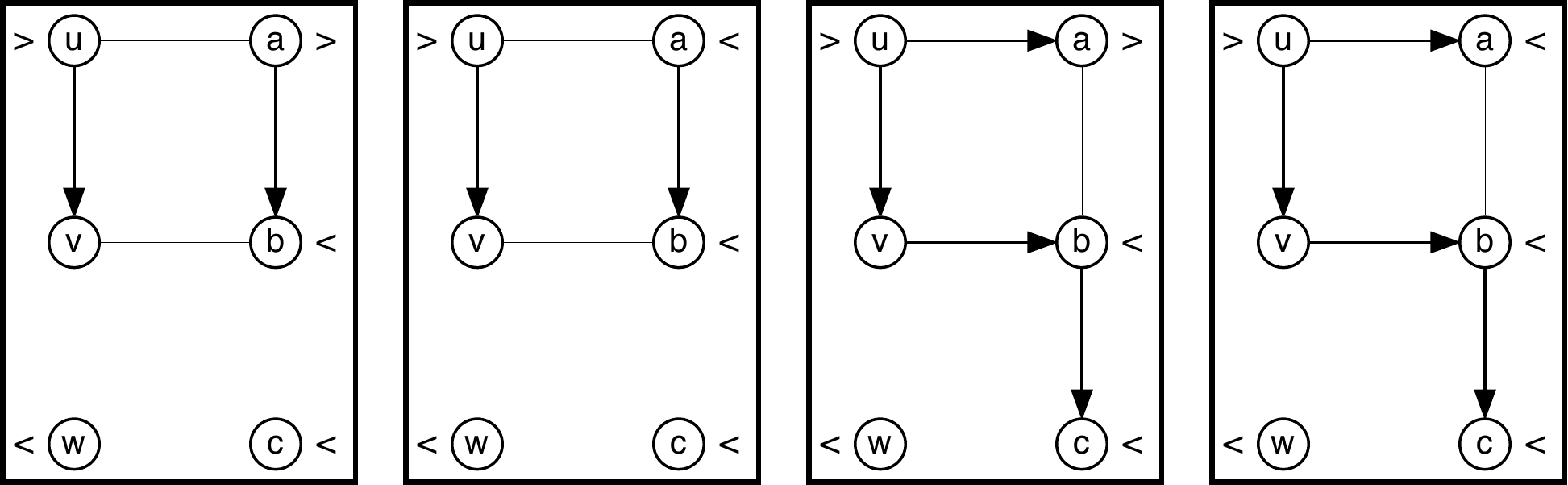}
	\hfill
	\includegraphics[width=.8\linewidth]{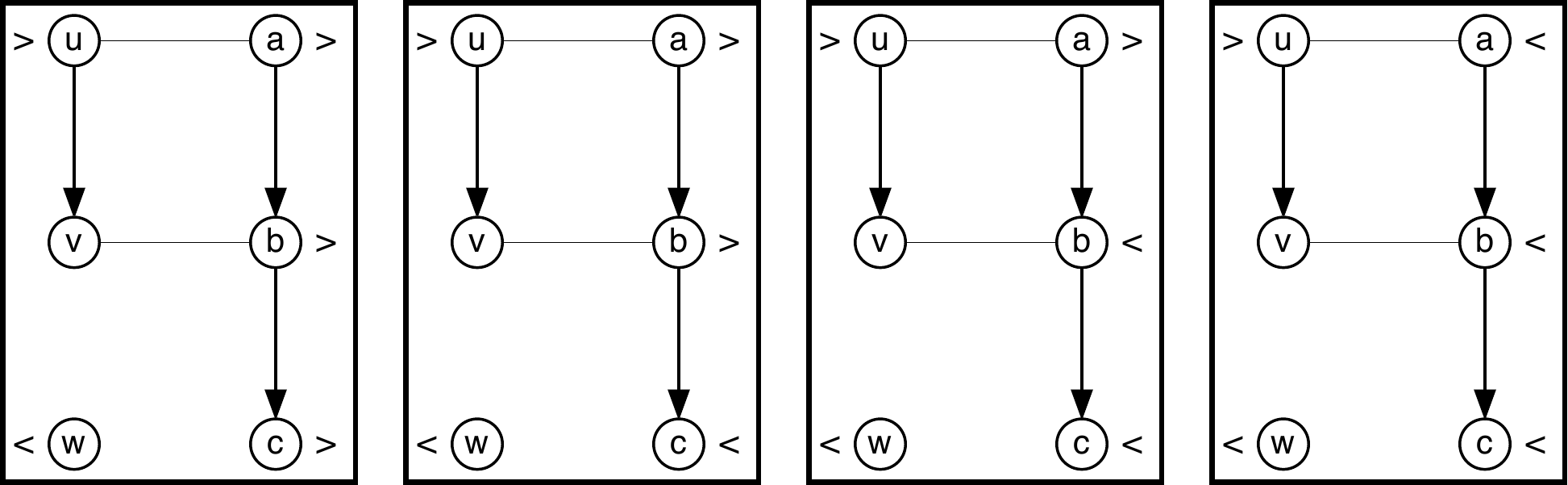}
	\hfill
	\includegraphics[width=.2\linewidth]{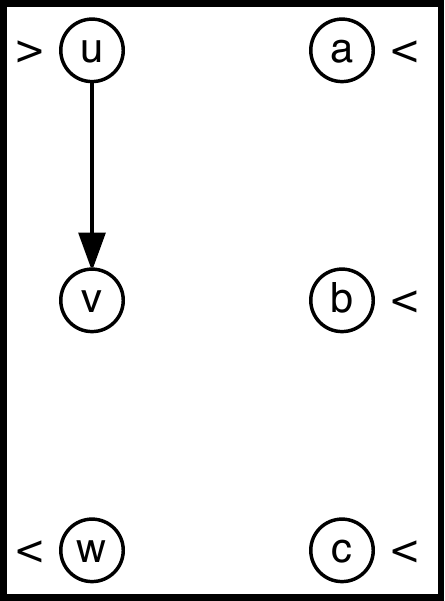}

	\caption[Cousin tree configurations for a node on the boundary.]{Cousin tree configurations for a node on the boundary: $a$ outside (2 cases); $a,b$ outside; $a,c$ outside; $a,w$ outside; $a,b,c$ outside; $a,b,w$ outside (violates assumption); $b$ outside; $b,c$ outside (4 cases); $b,c,w$ outside (4 cases); $c$ outside (5 cases); $c,w$ outside (5 cases); $w$ outside (6 cases); $a,b,c,w$ outside.}
	\label{fig:cousin_boundary_conditions}
	\end{figure}
\end{proof}

\end{document}